\documentclass[a4paper,10pt]{article}
\pdfoutput=1
\usepackage[table]{xcolor}
\usepackage{amsmath}
\usepackage{amssymb,tikz}
\usepackage{graphicx}
\usepackage{hyperref} 
\usepackage{fancybox} 
\usepackage{fancyhdr} 
\usepackage{geometry,hhline}       
\usepackage{multicol}
\usepackage{bbm}
\usepackage{multirow}
\usepackage{algorithm,algorithmic}

\newenvironment{proof}{\paragraph{Proof:}}{\hfill$\square$}
\newcommand{\ind}{\mathbbm{1}}
\usepackage[utf8x]{inputenc}
\usepackage[english]{babel}
\usepackage{float}	
\usepackage{booktabs}
\usepackage{subfig,color}
\usepackage[sort,nonamebreak, sectionbib,square]{natbib}
\usepackage[nottoc, notlof, notlot]{tocbibind}
\newtheorem{thm}{Theorem}[section]
\newtheorem{prop}[thm]{Proposition}

\newtheorem{exm}{Example}[section]
\newtheorem{rmk}{Remark}[section]
\allowdisplaybreaks
\newcommand{\indicator}[1]{\textbf{1}_{\lbrace {#1} \rbrace }}

\newcommand{\dd}{\mathrm{d}}
\title{Computation of Gaussian orthant probabilities in high dimension}
\author{James Ridgway\thanks{CREST-ENSAE and CEREMADE Universit\'{e} Dauphine, e-mail: james.ridgway@ensae.fr}}
\begin{document}
\maketitle
\begin{abstract}
We study the computation of Gaussian orthant probabilities, i.e. the probability that a Gaussian variable falls inside a quadrant. The Geweke-Hajivassiliou-Keane (GHK) algorithm \citep{Geweke1991,Keane1993,Hajivassiliou1996,Genz1992}, is currently used for integrals of dimension greater than $10$. In this paper we show that for Markovian covariances GHK can be interpreted as the estimator of the normalizing constant of a state space model using sequential importance sampling (SIS). We show for an AR(1) the variance of the GHK, properly normalized, diverges exponentially fast with the dimension. As an improvement we propose using a particle filter (PF). We then generalize this idea to arbitrary covariance matrices using Sequential Monte Carlo (SMC) with properly tailored MCMC moves. We show empirically that this can lead to drastic improvements on currently used algorithms. 

We also extend the framework to orthants of mixture of Gaussians (Student, Cauchy etc.), and to the simulation of truncated Gaussians.
\end{abstract}
\section{Introduction}
There are many applications where computing an orthant probability in high dimension with respect to a Gaussian or Student distribution is an issue of interest. For instance it is common in statistics to compute the likelihood of models, where we observe only an event with respect to multivariate Gaussian random variables. In Econometrics, the multivariate probit model \citep{Train2009}, where we observe a decision among $J$ alternative choices each of them corresponding to a Gaussian utility, is commonly studied. It can be written as an orthant problem. Other such models are the spatial probit \citep{LeSage2011} and Thurstonian models \citep{Yao1999}. Other applications than direct modelization can be found, such as multiple comparison tests \citep{Hochberg1987}, where the integration is done with respect to a Student (see \cite{Bretz2001} for an example). Orthant probabilities are also of interest in other fields than statistics, i.e. stochastic programming \citep{Prekopa1970}, structural system reliability \citep{Pandey1998}, engineering, finance, etc.

The problem at hand is the computation of the integral, 
\begin{equation}
	\int_{\left[\textbf{a},\textbf{b}\right]} \left(2\pi\right)^{-\frac{d}{2}}\left|\Sigma\right|^{-\frac12}\exp\left(-\frac12 (y-m)^t\Sigma^{-1} (y-m)\right)dy.
\label{eq:orthant}
\end{equation}
where $\textbf{a},\textbf{b} \in \mathbb{R}^d$. The Student case will be written as a mixture of the above integral with  an inverse Chi-square (see Section \ref{sec:Student}). 

Many algorithms have been proposed to compute \eqref{eq:orthant}; for a review see \cite{Genz2009}. They can be divided into two groups. The first are numerical algorithms to deal with small dimensional integrals. In dimension $3$ there exist algorithms \citep{Genz2009} where after  sphericization, such that the Gaussian has an identity covariance matrix, one applies recursively numerical computations of the error function. For higher dimensions than three \cite{Minwa2003} propose to express orthant probabilities as differences of orthoscheme probabilities, where an orthoscheme is \eqref{eq:orthant} with correlation matrix $\Omega=(\omega_{ij})$ satisfying $\omega_{ij}=0\quad\forall i,j \quad \vert i-j\vert>1$. This  can be easily computed by recursion. However the decomposition in orthoscheme probabilities has factorial complexity. The second group of algorithms is Monte Carlo based and may be used for dimensions higher than 10. In particular GHK due to \cite{Geweke1991}, \cite{Keane1993} and \cite{Hajivassiliou1996} and conjointly to \cite{Genz1992}, has been widely adopted for the applications described above.

In this paper we show that in the case of Markovian covariances (i.e. covariances that can be written as those of Markovian processes), the GHK algorithm estimates the normalizing constant of a state space model (SSM), using  sequential importance sampling (SIS) with optimal proposal. We show in addition for a first order autoregressive process (henceforth AR(1)) that the normalized variance diverges exponentially fast.

To avoid this behavior we propose to use a particle filter. We extend this methodology to the non Markovian case by using Sequential Monte Carlo (SMC). SMC allows additional gain in efficiency by considering different MCMC moves and proposals. In addition the algorithm is adaptive and simplifies automatically to the GHK if the integral is simple enough. In our numerical experiments we find a substantial improvement.

We start by reviewing the existing GHK algorithm (Section \ref{sec:GHK}), we then discuss the algorithm's behavior for Markovian covariance matrices and propose an extension to higher dimensions (Section \ref{sec:Markovian}). In Section \ref{sec:non-Markovian} we extend this proposal to arbitrary covariance matrices. We propose some extensions for the simulation of truncated distributions and for other distributions (\ref{sec:Extensions}). Finally we present some numerical results and conclude (Sections \ref{sec:NumAna} and \ref{sec:ccl}).

\paragraph{Notations}
 For any vector $x\in\mathbb{R}^p$ for $i\leq p$ we write $x_{<i}\in\mathbb{R}^i$ for the vector of the $i-1$ first components, and we take $a:b=\{a,\cdots,b\}$. We let $x_{<1}=\emptyset$, and also write $x_{i:j}$ for the vector $(x_i,x_{i+1},\cdots,x_{j})$;
 $\Phi,\varphi$ are respectively the $\mathcal{N}(0,1)$ Gaussian cdf and pdf, we write $\varphi(x\vert A)$ for the pdf, $\frac{\varphi(x)}{\Phi(A)}\ind_{A}(x)$, of a Gaussian truncated to the set $A \subset \mathbb{R}$ evaluated in $x$. We will also abuse notation and use $\Phi(A)$ to denote the probability of a set when $A\subset\mathbb{R}$. For instance $\Phi(\left[a,b\right])=\Phi(b)-\Phi(a)$.

\section{Geweke-Hajivassiliou-Keane (GHK) simulator}
\label{sec:GHK}
From now on to simplify notations, and without loss of generality, we limit ourselves to the study of the following multidimensional integral:
\begin{equation}
	F(\textbf{a},\textbf{b},\Sigma)=\int_{\left[\textbf{a},\textbf{b}\right]} (2\pi)^{-\frac{d}2}\vert\Sigma \vert^{-\frac12}\exp\left(-\frac12y^t\Sigma^{-1} y\right)dy
\label{eq:orthant2}
\end{equation} 
with $\textbf{a},\textbf{b}\in \mathbb{R}^d$. Note that the extension to integrals where some components of the vectors $\textbf{a},\textbf{b}$ are respectively $-\infty$ and $\infty$ is direct. 

Let $\Gamma$ be  the Cholesky decomposition of $\Sigma$, i.e. $\Sigma=\Gamma\Gamma^t$ with $\Gamma=(\gamma_{ij})$, $\gamma_{ii}>0$ and $\gamma_{ij}=0$ if $j>i$. We can write the previous equation after the change of variable $\eta=\Gamma^{-1} y$ for which $d\eta=\vert\Gamma\vert^{-1} dy$:
\[
	F(\textbf{a},\textbf{b},\Sigma)=\int_{ \textbf{b} \geq\Gamma\eta\geq \textbf{a}}(2\pi)^{-\frac{d}2}\exp\left(-\frac12\eta^t \eta\right)d\eta,
\]
the $i$-th truncation being such that $\frac1{\gamma_{ii}}\left(a_i-\sum_{j=1}^{i-1} \gamma_{ij} \eta_j\right)\leq \eta_i \leq \frac1{\gamma_{ii}}\left(b_i-\sum_{j=1}^{i-1} \gamma_{ij} \eta_j\right)$, from the positivity of the $(\gamma_{ii})$. Thus we can write:
$$
F(\textbf{a},\textbf{b},\Sigma)=\int \prod_{i=1}^d \varphi(\eta_i)\indicator{B_i(\eta_{<i})}(\eta_i)d\eta_{1:d}=\int \prod_{i=1}^d\Phi\left(B_i(\eta_{<i})\right) \varphi(\eta_i\vert B_i\left(\eta_{<i})\right)d\eta_{1:d},
$$
where the set $B_i(\eta_{<i})=\lbrace \eta_i : \frac1{\gamma_{ii}}\left(a_i-\sum_{j=1}^{i-1} \gamma_{ij} \eta_j\right)\leq \eta_i \leq \frac1{\gamma_{ii}}\left(a_i-\sum_{j=1}^{i-1} \gamma_{ij} \eta_j\right) \rbrace$ is an interval.

The GHK algorithm is an importance sampling algorithm based on this structure. It proposes particles distributed under $\prod_{i=1}^d \varphi(\eta_i\vert B_i\left(\eta_{<i})\right)$ and evaluates the average of the weights $w^n=\prod_{i=1}^d\Phi\left(B_i(\eta^n_{<i})\right)$. The algorithm is described in pseudo-code in  Alg. \ref{alg:GHK}.

\begin{algorithm}[H]
\caption{GHK simulator}
\begin{algorithmic}
\FOR{$m\in 1:M$}
\STATE \underline{Sample}: $\eta^m_{1:d}\sim \prod_{i=1}^d \varphi(\eta_i \vert B_i(\eta_{<i}))$ 
\STATE \underline{Weights$^\star$}: $w^m= \prod_{i=1}^d\Phi\left(B_{i}(\eta^m_{<i})\right)$
\ENDFOR
\RETURN $\frac1M \sum_{i=1}^M w^i$ 
\end{algorithmic}
\label{alg:GHK}
\label{fig:highdim}
\hrule
\vspace*{1mm}
\begin{minipage}{15cm}
\footnotesize{$^\star$Recall that $\Phi(B_i(\eta^m_{<i})$ can be computed as a difference of two one dimensional cdf for the truncation defined above.}
\vspace*{1mm}
\end{minipage}
\end{algorithm}

Algorithm \ref{alg:GHK} outputs an unbiased estimator of integral \eqref{eq:orthant}. 

To generate truncated Gaussian variables the usual approach in the GHK simulator is to use the inverse cdf method. We follow this approach in the rest of the paper except where stated otherwise. When the numerical stability of the inverse cdf is an issue we will use the algorithm proposed in \cite{Chopin2011b}.

In the next section we will study with more care the case where the covariance matrix of the underlying  Gaussian vector has a Markovian structure. 

\section{The Markovian case}
\label{sec:Markovian}

When the covariance matrix is Markovian, that is a matrix for which the inverse is tri-diagonal, the simulation step of Alg. \ref{alg:GHK} is the simulation of a Markov process $(x_{1:t})$. At time $t$ the weights depend on $x_{t-1}$ only. Let us take a lag $1$ autoregressive process (AR(1)) for the purpose of exposition, and study the probability of it being in some hyperrectangle $\left[\textbf{a},\textbf{b}\right]=\left[a_1,b_1\right]\times\cdots\times \left[a_T,b_T\right]$. 
The integral of interest is therefore:
\begin{equation}
\int \prod_{t=1}^T \indicator{\left[a_t,b_t\right]}(x_t)\varphi(x_t; \varrho_t x_{t-1},\sigma_t^2)dx_{1:T}.
\label{eq:Markov_lik}
\end{equation}
The GHK algorithm consists in sampling from the Markov process:
\[
x_t \vert x_{t-1} \sim \varphi\left(\varrho_t x_{t-1},\sigma_t^2\vert B_t(x_t)\right).
\] 
The matrix $\Sigma^{-1}$ is tridiagonal,  the weights at time $t$ are therefore $\Phi(\frac{b_t-\varrho_t x_{t-1}}{\sigma_t})-\Phi(\frac{a_t-\varrho_t x_{t-1}}{\sigma_t})$. Eq. \ref{eq:Markov_lik} can be seen as the likelihood of the state space model \citep{Cappe2005}:
\begin{align*}
&x_t\vert x_{t-1}\sim \varphi(x_t; \varrho_t x_{t-1},\sigma_t^2)\\
&y_t\vert x_{t}\sim \indicator{\left[a_t,b_t\right]}(x_t)
\end{align*}
where $(y_t)_t$ is observed. The GHK can be interpreted as a sequential importance sampler (SIS) using proposal $\varphi(\varrho_t x_{t-1},\sigma_t^2\vert B_t(x_t))$. 

\subsection{Toy example}
\label{sec:toy}
Let us specify a bit more the problem to simplify notation and show some properties of a thus defined algorithm.

Consider the problem of finding the probability that an AR(1),
 $$
 X_t=\varrho X_{t-1}+\varepsilon_t,\qquad \vert\varrho\vert<1
 $$
is inside the hyper-cube $\left[0,b\right]\times \cdots\times\left[0,b\right]$, for some $b>0$. We have set $\sigma=1$, $a=0$ and $\varrho$ a constant.

 The GHK algorithm consists in this case in simulating the above Markov chain constrained to  $\left[0,b\right]$ and in computing under this distribution the products of the weights $\prod_{t=1}^T \left[\Phi(b-\varrho X_t)-\Phi(-\varrho X_t)\right]$.
The simulations are therefore generated by the Markov probability kernel
\begin{equation}
P^b(x,dy)=\frac{\varphi(y;\varrho x,1)}{\Phi(b-\varrho x)-\Phi(-\varrho x)}\ind_{\left[0,b\right]}(y)dy, \quad \vert\varrho\vert<1.
\label{eq:transition}
\end{equation}

For this model we have the following proposition:

\begin{prop}
For the Markov model defined by $\eqref{eq:transition}$, the normalized square product of weights of the normalizing constant has the following behavior:
\begin{equation}
	\liminf_{T\rightarrow\infty}\left\lbrace\mathbb{E}\left[\left(\frac{\prod_{t=1}^T\big(\Phi\left(b-\varrho X_t\right)-\Phi\left(-\varrho X_t\right)\big)^2}{\exp\lbrace2T\mathbb{E}_\pi\log\left(\Phi(b-\varrho X)-\Phi(-\varrho X)\right)\rbrace}\right)\right]\right\rbrace^{\frac1{\sqrt{T}}}>\exp\left\lbrace\mathbb{V}_\pi\left[\psi(X)\right]+\tau\right\rbrace,
\label{eq:prop}
\end{equation}
where subscript $\pi$ denotes integration with respect to the invariant distribution of $P^b(x,\dd y)$, the other expectation is taken relatively to the Markov chain $(X_t)$, and $\psi:x\mapsto\log\left(\Phi(b-\varrho x)-\Phi(-\varrho x)\right)$, $\tau=2\sum_{k=1}^\infty \text{cov}(X_0,X_k)$. 
\end{prop}
\begin{proof}
A detailed proof is given in appendix \ref{app:proof}.
\end{proof}

Under $V$-Uniform ergodicity, that follows from our proof, the denominator is the square of the limit of the product of weights and can be interpreted as a scaling factor. Thus the result above shows that this renormalized squared estimator diverges exponentially fast as the dimension of the integral increases.  

\begin{rmk}
\label{rmk:lognorm}
In the course of the proof we showed that the normalizing constant has a log-normal limiting distribution, resulting in a skewed distribution. We expect that the distribution of the estimator will have its mode away from the expected value resulting in some apparent bias. In fact one can show that the normalized third order moment will also grow exponentially.  
\end{rmk}

GHK has quadratic complexity however we can show that for at least one covariance structure the variance diverges exponentially fast. This fully justifies the use of an algorithm of higher computational complexity. In the following section we propose a natural extension to deal with this issue in the Markovian case.

\subsection{Particle filter (PF)}

PF is a common extension of SIS that corrects the weight degeneracy problem. The solution brought by particle filtering \citep{Gordon1993} is to use a resampling step, i.e. to kill those particles with low weights and to replicate those with high contribution. At time $t$ one resamples the particles by sampling from the distribution $\sum_{m=1}^M W^m_t \delta_{x_t^m}(\dd x)$ where $W_t^m$ stands for the $m$-th renormalized weight at time $t$, and $\delta_x(dx^\prime)$ the Dirac measure in $x$. All the weights are then set to one. 

We use an adaptive version of this algorithm where the resampling step is triggered only when the ESS of the weight is lower than some threshold, where the ESS is defined as
\[
	\frac{\left(\sum_{i=1}^N w_i\right)^2}{\sum_i w_i^2}\in\left[1,M\right],
\] 
and indicates the number of draws from the  independent distribution to obtain the same variance. Note that it is closely related to the inverse of equation \eqref{eq:prop}, hence we expect that without resampling it goes to zero with exponential speed.

We define the state space model:
\begin{align*}
x_t\vert x_{t-1} &\sim g_t(x_t\vert x_{t-1}),\\
y_t\vert x_{t}& \sim f_t(y_t\vert x_{t})
\end{align*}

One can use a PF to compute the likelihood of such model,
\[
	L(y_{1:T})=\int \prod_{t=1}^Tg_t(x_t\vert x_{t-1})f_t(y_t\vert x_t)g_0(x_0)dx_{0:T}
\]

A PF with proposal distribution $q_t(x_t,x_{t-1})$ is described in Alg. \ref{alg:PF}. Our application corresponds to the special case where:
\begin{align*}
g_t(x_t\vert x_{t-1})=\varphi(x_t),\qquad
f_t(y_t\vert x_t,x_{t-1})=\indicator{B_t(x_{<t})}(x_t),\qquad
q_t(x_t\vert x_{t-1})=\varphi(x_t\vert B_t(x_{<t})),\qquad
\end{align*}
where the set $B_t(x_{<t})$ depends on $x_{t-1}$ only.

The proposal thus defined corresponds to the optimal one \citep{Doucet2000}, that is the distribution proportional to $f_t(y_t\vert x_t)g_t(x_t \vert x_{t-1})$ in our case proportional to $\varphi(x_t )\indicator{B_t(x_{<t})}(x_t)$ hence the truncated Gaussian. The weights are given by the normalizing constant $\int f_t(y_t\vert x_t)g_t(x_t \vert x_{t-1})\dd x_t$, in our case $\Phi(B_t(x_{<t}))$.

To resample we propose to use systematic resampling \citep{Carpenter1999} (for other approaches see \cite{Douc2005}). Systematic resampling is described in Algorithm \ref{alg:Resample} (Appendix \ref{app:resampling}).

 \begin{algorithm}[h!]
 \caption{Particle Filter}
\begin{algorithmic}
\STATE \underline{Input}: $M$ the number of particles
\STATE \underline{Sample}: Sample $x^i_0\sim g_0(.)$
\FOR{$t=1:T-1$}
\IF{$ESS<\eta^\star$}
\STATE $Z\leftarrow Z\times \lbrace\frac1M \sum_{i=1}^Mw_t^i\rbrace$
\STATE Resample  $a^j_t \sim \sum_i \frac{w^i_t}{\sum_j w^j_t} \delta_{i} $ using algorithm \ref{alg:Resample}, set $w^j_t\leftarrow 1$
\ELSE
\STATE $a_t^{1:M}=1:M$
\ENDIF
\STATE Sample $x^i_{t+1}\sim q_{t+1}(.\vert x^{a^i_t}_{t})$
\STATE Set $w^i_{t+1}\leftarrow w^i_t \frac{f_{t+1}(y_{t+1}\vert x^i_{t+1})g_{t+1}(x^i_{t+1}\vert x^{a^i_t}_{t})}{q_{t+1}(x^i_{t+1}\vert x^{a^i_t}_{t})}$
\ENDFOR
\RETURN $Z\times\frac1M\sum_{i=1}^M w_T^i$
 \end{algorithmic}
\label{alg:PF}
\end{algorithm}

The particle filter thus defined outputs an unbiased estimator of the likelihood \cite{DelMoral1996}, and thus the orthant probability in our case. 

Note that the output of Algorithm \ref{alg:PF} is of the form of a product of terms smaller than one, in our case those terms can be very small and lead to numerical issues. One way of dealing with this issue is to rewrite all the algorithm in log scale.

\begin{rmk}
	As we are here in the special case of being able to sample from the optimal distribution  (as shown in Section \ref{sec:Markovian}) one could resort to the auxiliary particle filter (APF, \cite{Pitt1999}). In fact in this special case the algorithm amounts to exchanging the resampling step and the move step of the particle filter. We tested this approach on some Markov processes and observed no improvements in term of variance on repeated draws. 
\end{rmk}
\begin{exm}
	We can show that the previous process (Section \ref{sec:toy}) benefits from resampling when the ESS goes beneath a given level.  

Figure (\ref{fig:AR1}) shows  that the GHK algorithm's variance increases more quickly as compared to the PF (that seem to have some stable variance on the considered dimension). In addition the distribution of the GHK estimator seem to be skewed towards smaller values as $T$ increases. This results in some bias on the last boxplot. As described in remark \ref{rmk:lognorm} this behavior is due to the log-Normal limiting distribution of the output of the algorithm. The skewness coefficient increases exponentially with $T$.

\begin{figure}[H]
\begin{center}
\begin{tabular}{cc}
\subfloat[variance]{\includegraphics[scale=0.4]{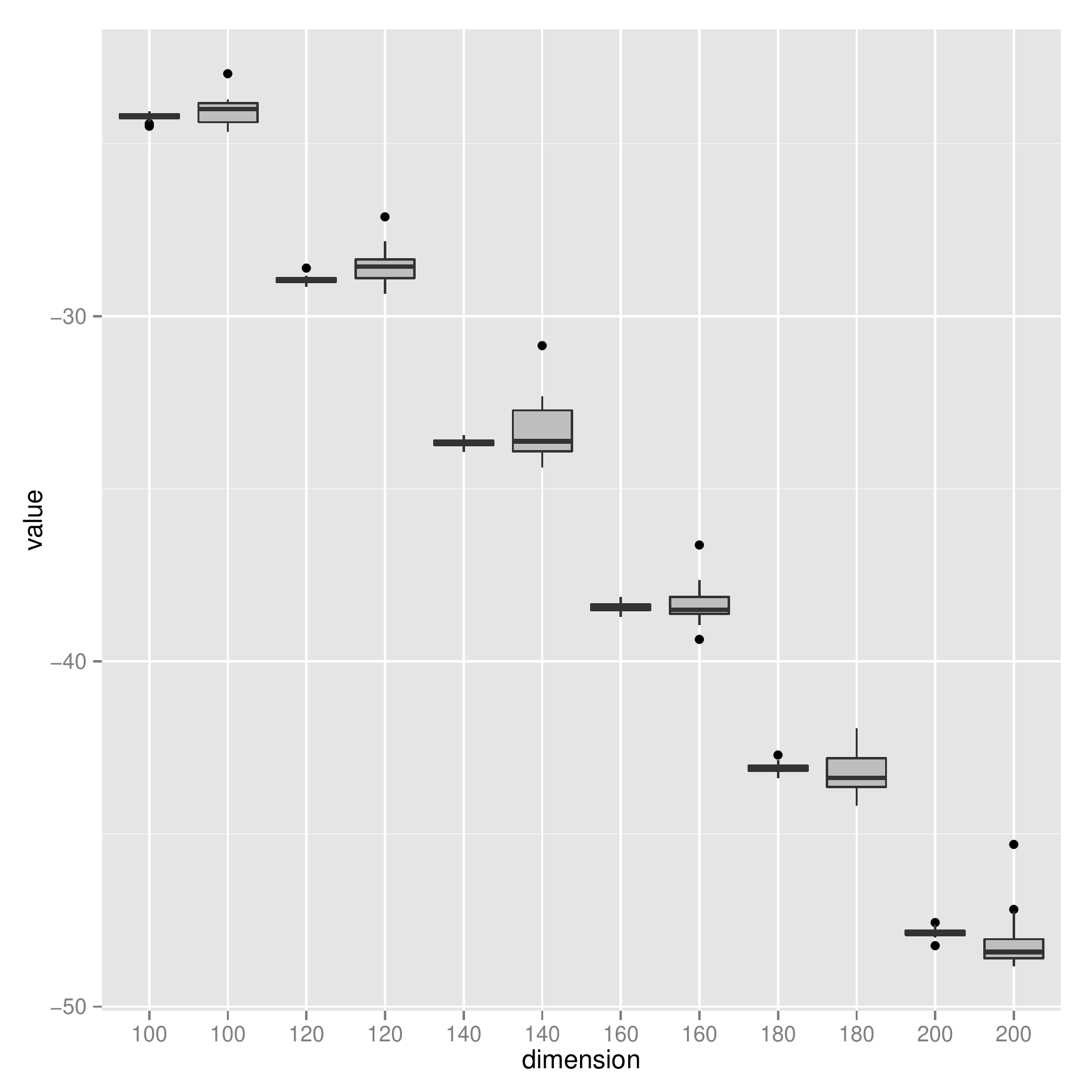}}
\subfloat[ESS]{\includegraphics[scale=0.4]{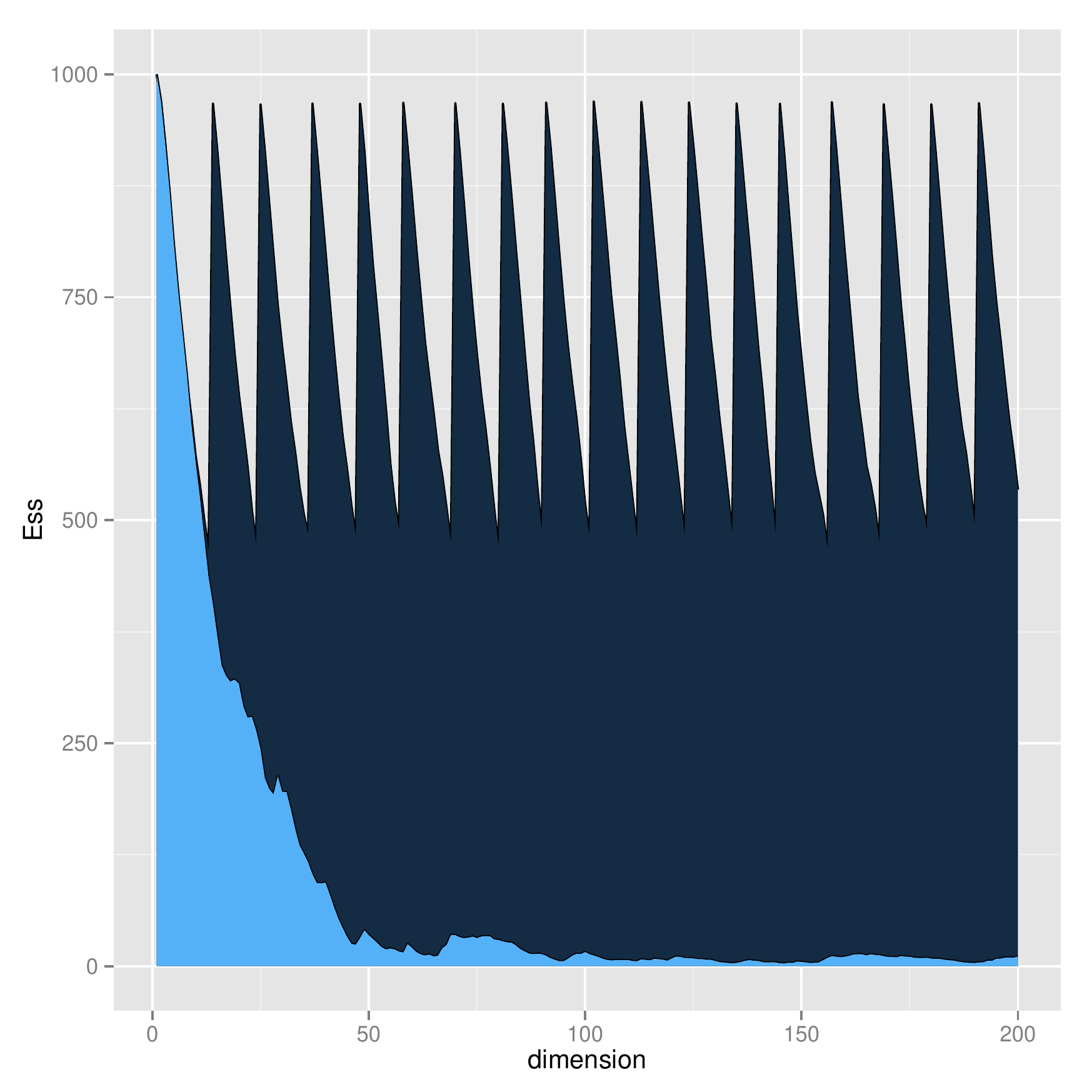}}
\end{tabular}
\caption{Estimates of Orthant probabilities by Particle Filter}
\label{fig:AR1}
\begin{minipage}{15cm}
\footnotesize{ Estimation of the log probability that an AR(1) process (defined previously) with $\varrho=0.7$ has all its component in $\left[0,15\right]$. GHK sampler (grey) and PF (white) on various dimension from 100 to 200. On the right panel the two ESS for dimension 200. On both cases $M$ is set to $1000$.}
\vspace*{3mm}
\hrule
\end{minipage}
\end{center}
\end{figure}
\end{exm}

\subsubsection{Thurstonian Model}

Thurstonian models arise in Psychology and Economics \citep{Yao1999} to describe the ranking of $p$ alternatives by $n$ individuals (referred to as judges). 

Suppose that we observe the rank $r_i=(k_{1i},\cdots ,k_{p,i})$ of some $p$ independent Gaussian random variables,
\[
x_{i,j}=\beta_j+\sigma \varepsilon_{i,j},
\]
where $\varepsilon_{i,j}\substack{i.i.d \\ \sim} \mathcal{N}(0,1)$. The likelihood of one observation is an orthant probability:
\begin{equation}
\label{eq:thurston}
\mathbb{P}_\theta\lbrace X_p>\cdots >X_1\rbrace=\int \prod_{i=1}^p\indicator{x_i>x_{i-1}}\varphi(x_i\vert \beta_j,\sigma^2)dx_{1:p}
\end{equation}
with the convention that $X_0=-\infty.$

This model is similar to the previous one but with $\rho=1$. 

\begin{figure}[H]
\begin{center}
\begin{tabular}{cc}
\subfloat[Estimates]{\includegraphics[scale=0.4]{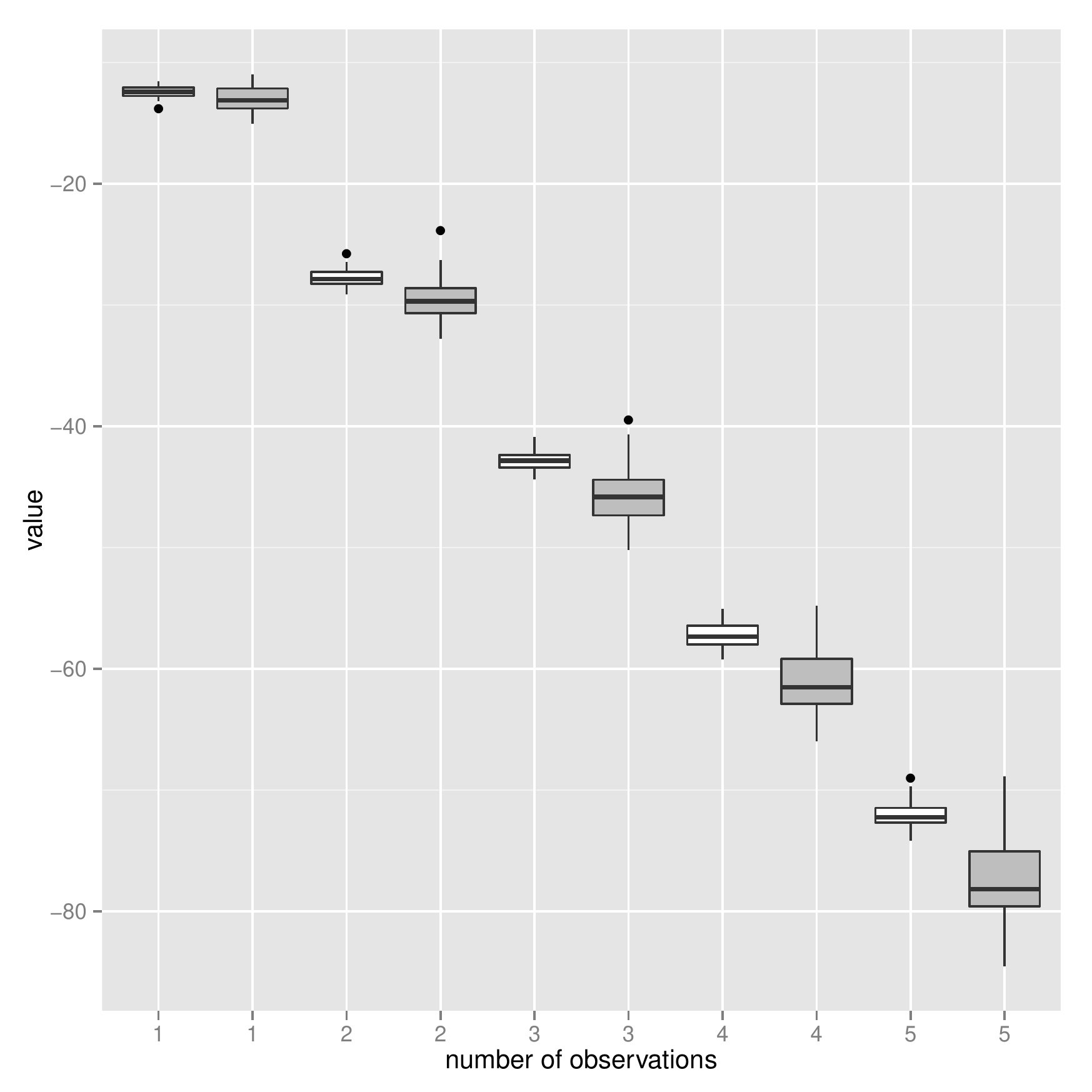}}&
\subfloat[ESS]{\includegraphics[scale=0.4]{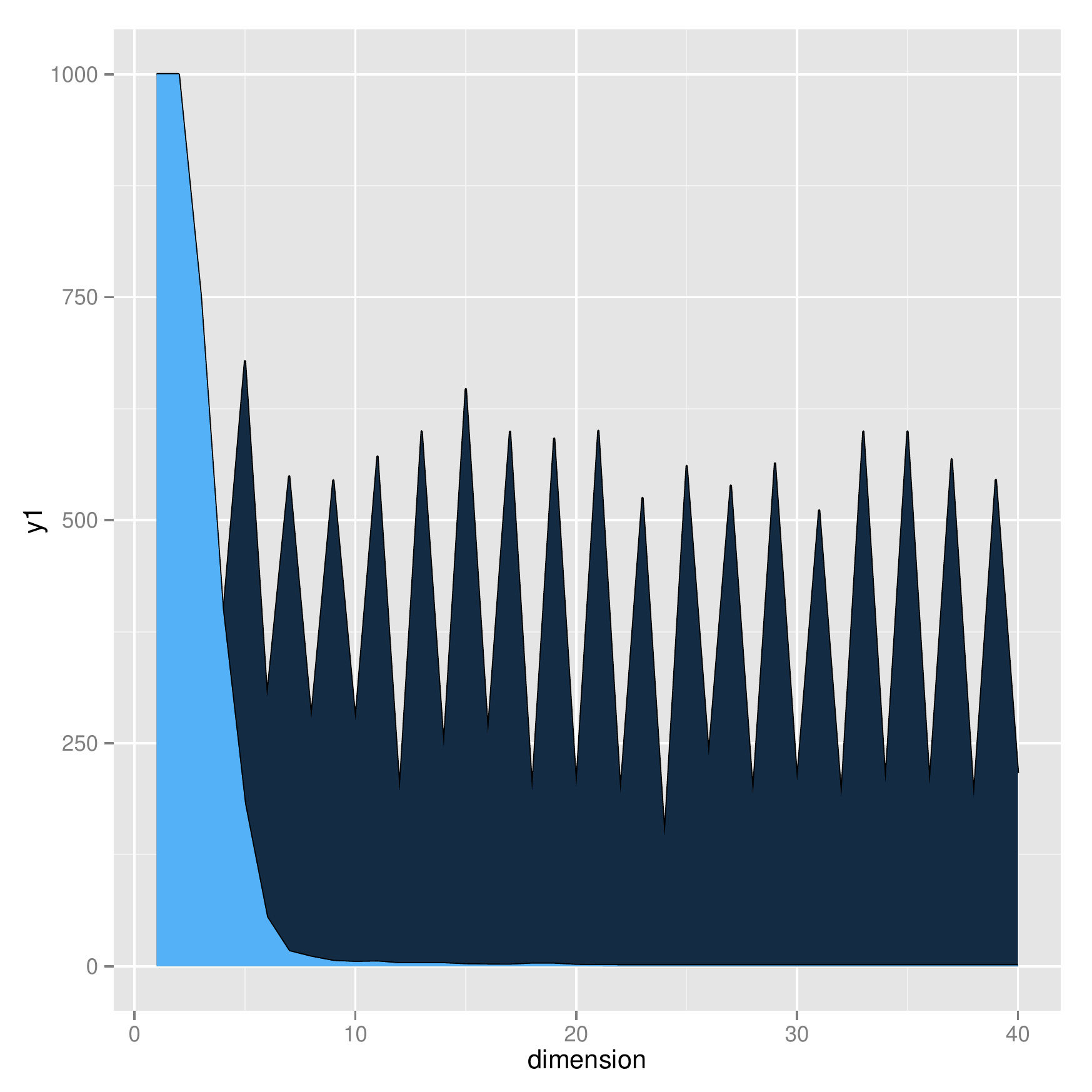}}
\end{tabular}
\caption{Estimates of the likelihood of a Thurstonian model by Particle Filter and by GHK}
\label{fig:thurstone}
\begin{minipage}{15cm}
\footnotesize{ Estimation of the likelihood of a Thurstoninan model with $p=10$ and the number of observations is ranging from $1$ to $5$ the PF (white) and the GHK (grey). The threshold ESS is set to $0.5M$ and the number of particles is set to $M=1000$. The right panel shows the ESS of both algorithms for $T=1$ and $p=40$.}
\vspace*{3mm}
\hrule
\end{minipage}
\end{center}
\end{figure}

We find that the likelihood is estimated with smaller variance. In addition, because of the heavy tail distribution of the GHK simulator's output we observe a bias (see Figure \ref{fig:thurstone}). Again one can explain the strong observed bias by remark \ref{rmk:lognorm} and the fact that we do not replicate enough the experiment to observe the tail of the distribution. In addition as suggested above the ESS of the GHK seems to decrease exponentially fast to zero.

From this observation we could apply this algorithm to perform inference by using Particle MCMC \citep{Andrieu2010}, where this estimation of the likelihood can be plugged in a Random walk Metropolis Hastings and still target the appropriate distribution. 

\section{Non Markovian case}
\label{sec:non-Markovian}
For more general covariances we propose to use Sequential Monte Carlo (SMC) \citep{DelMoral2006}. As previously we will base the algorithm on the proposal of GHK, increasing the dimension of the problem at each time step. However we now have an additional degree of freedom: the order in which we incorporate the variables. In the following section we study an approach to ordering the variables.

\subsection{Variable ordering}
We follow \cite{Gibson1994} in ordering the variables from the most difficult to the simplest, where difficult constraints are considered to be the one that impact the most the probability. 

However we cannot evaluate exactly the probabilities as it is our final goal. Instead \cite{Gibson1994} propose to replace the simulations by the expected value of the truncated Gaussian.

The algorithm starts by choosing the first index $i_1$, and defining $\eta_1$ as follows:
\[
	i_1=\arg \min_{1\leq k\leq T} \Phi\left(\left[\frac{a_k}{\gamma_{kk}},\frac{b_k}{\gamma_{kk}}\right]\right),\qquad \eta_1=\frac1{\Phi\left(\left[\frac{a_{i_1}}{\gamma_{i_1 i_1}},\frac{b_{i_1}}{\gamma_{i_1 i_1}}\right]\right)}\int_{\left[\frac{a_{i_1}}{\gamma_{i_1 i_1}},\frac{b_{i_1}}{\gamma_{i_1 i_1}}\right]} \eta\varphi(\eta)d\eta
\]
i.e. the smallest possible probability that the Gaussian will be in $\left[a_k,b_k\right]$. This enables an approximation of the next probability as a function of $i_2$.
\[
	i_2=\arg \min_{2\leq k\leq T} \Phi\left(\left[\frac1{\tilde{\gamma}_{kk}}(a_k-\tilde{\gamma}_{1,k}\eta_1),\frac1{\tilde{\gamma}_{kk}}(b_k-\tilde{\gamma}_{1,k}\eta_1)\right]\right).
\]
where $(\tilde{\gamma}_{ij})=\tilde{\Gamma}$ is the Cholesky decomposition of the matrix after substituting the first and the $i_1$th variable.  

We end up with the desired vector $(i_1,\cdots,i_T)$ that gives us the order in which to choose the covariances and truncation points. The algorithm is summed up by Alg. \ref{alg:VO} in appendix \ref{app:VO}. The algorithm has quadratic time complexity, however its cost is negligible as compared to the subsequent Monte Carlo algorithm.

We show the use of the reordering in moderate dimensions (50 and 60) on the GHK simulator. This is already a great improvement especially as the dimension increases. Figure (\ref{fig:order}), shows boxplots of 50 repetitions of the GHK for both ordered (white) and non-ordered (grey) inputs. 

\begin{figure}[H]
\begin{center}
\begin{tabular}{cc}
\subfloat[Dimension 50]{\includegraphics[scale=0.3]{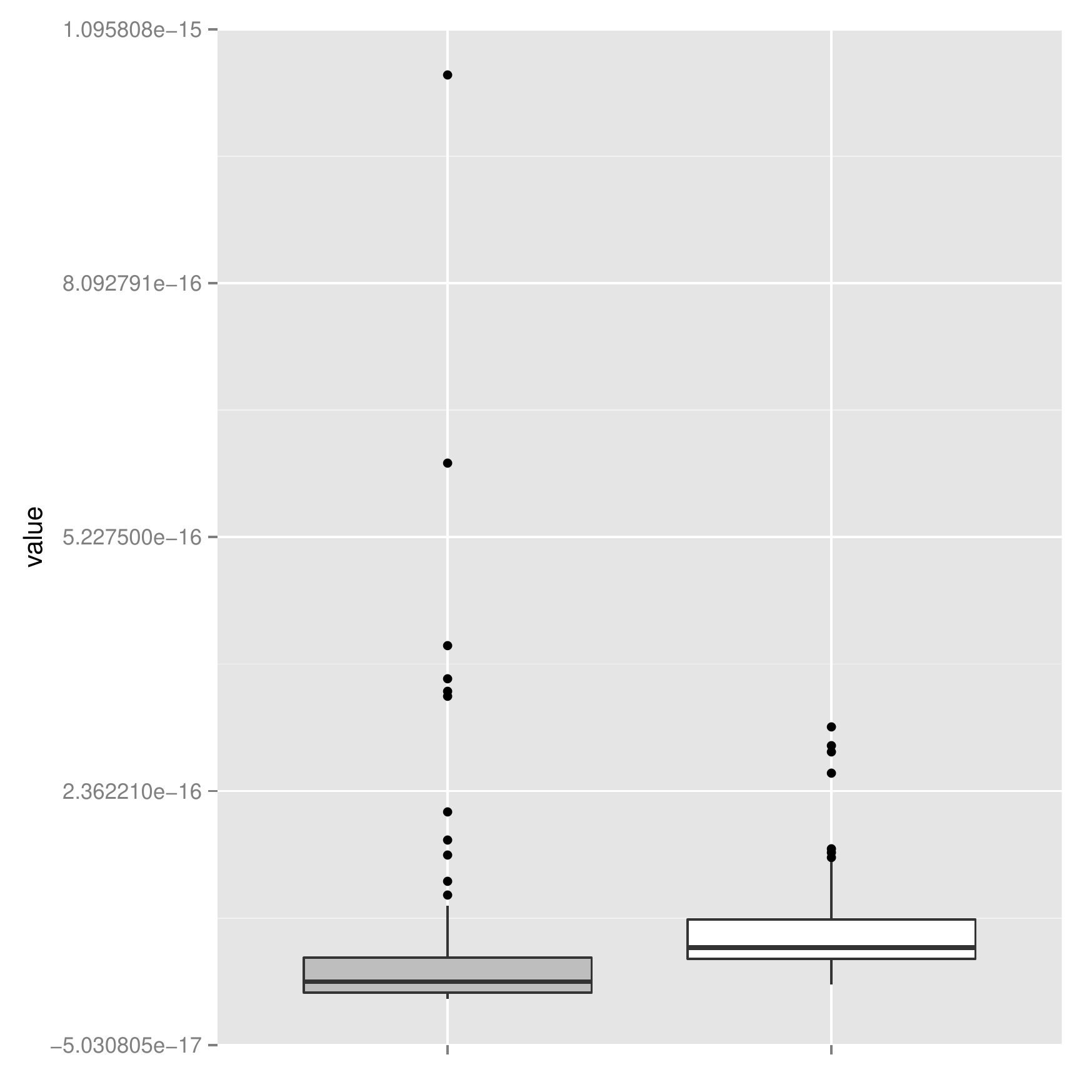} } &
\subfloat[Dimension 60]{\includegraphics[scale=0.3]{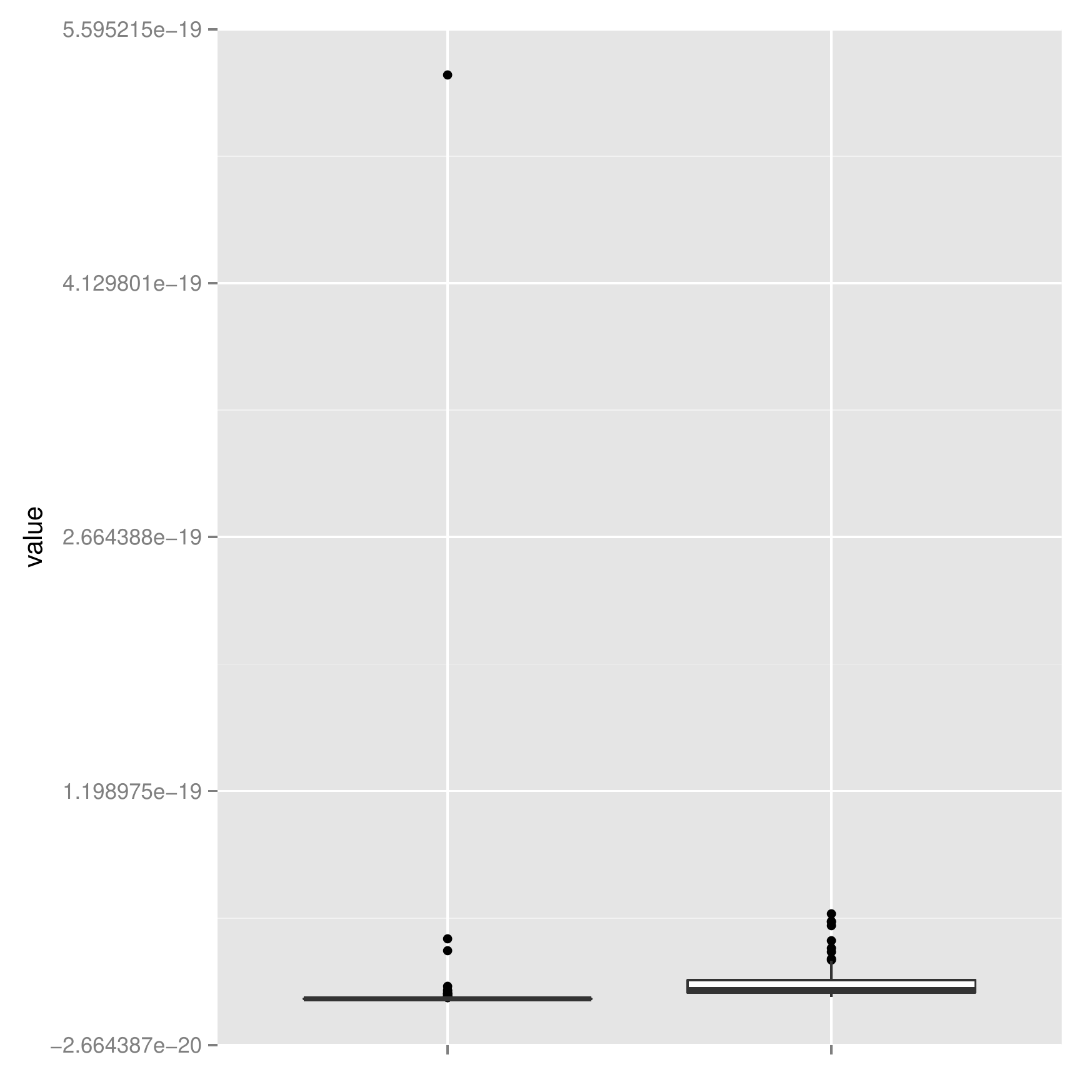} }
\end{tabular}
\caption{Estimates of Orthant probabilities with (white) and without (grey) variable ordering}
\label{fig:order}
\begin{minipage}{15cm}
	\footnotesize{Covariance matrices generated from random samples with heavy tails (see Section \ref{sec:NumAna}). In both case we use a GHK simulator with variable ordering (white), without (grey). The various dimension are simulated with the same algorithm and same seed, such that the small ones are subsets of the others. When the variables are not ordered we observe some outliers, and this phenomenon is reduced with \cite{Gibson1994}'s algorithm.}
\vspace*{3mm}
\hrule
\end{minipage}
\end{center}
\end{figure}

From dimension 50 and upwards we start observing some skewed distributions for the GHK estimator as noted in remark \ref{rmk:lognorm}. The phenomenon seems to be reduced by ordering. 

This effect is relatively dangerous as some draws depart a lot from the mean value. The ordering will be used on all examples from now on to reduce the variance. In Section \ref{sec:NumAna} we have empirical evidence that using an appropriate move step deals with this tail effect, in our examples. 

We have shown that in the particular case of Markov processes we can strikingly benefit from the use of resampling. In the next section we attempt to generalize our finding to a broader range of problems.
\subsection{A sequential Monte Carlo (SMC) algorithm}
\label{sec:smc}
The algorithm discussed in Section \ref{sec:Markovian} can be generalized to non Markovian Gaussian vectors by applying the SMC methodology. Define the following sequence of distribution:
\begin{equation}
	\label{eq:target}
\pi_t(\eta_{1:t})=\frac{\gamma_n(\eta_{1:t})}{Z_t},\qquad\gamma_t(\eta_{1:t})=\prod_{i=1}^t\varphi(\eta_i)\indicator{B_i(\eta_{<i})},
\end{equation}
indexed by $t$, where the unnormalized quantity $\gamma_t(\eta_{1:t})$ is our target integrand. We thus want to compute an estimator of $Z_t$ for a given $t$, 
$$
Z_t=\int \prod_{i=1}^t\varphi(\eta_i)\indicator{B_i(\eta_{<i})} \dd \eta_{1:t}.
$$

SMC samplers are a class of algorithms that generalize particle filters to non dynamic problems (\cite{Neal2001},\cite{Chopin2002}, \cite{DelMoral2006}). Their aim is to sample from a sequence of measures $(\pi_t)_t$ where $\pi_0$ is easy to sample from and $\pi_T$ is our target. The algorithm works by moving from one target to the other by importance sampling, and avoid degeneracy of the weights by resampling if the ESS falls bellow a threshold. In the case of the GHK the sequence of distribution consist of adding a dimension at each step. To ensure particle diversity after resampling the particles are moved according to a MCMC kernel targeting the current distribution. This is the most computationally expensive step. Different alternatives are described in the next section. 

The main steps are described in Algorithm (\ref{alg:SMC}) bellow:
\begin{algorithm}[H]
\caption{SMC for orthant probabilities}
\begin{algorithmic}
\STATE \textbf{Input} $\Gamma$,$a$, $M$, $\alpha^\star$, Set $Z\leftarrow 1$
\STATE Each computation  involving $m$ is done $\forall m \in 1:M$
\STATE \textbf{Init} $\eta^m_1\sim \varphi(.)\textbf{1}_{B_k(\eta_{<1})}$, and $w^m_1=1$
\FOR{$t\in 1:T-1$}
\STATE  At time $t$ the weighted system is distributed as $(w^m_t,\eta^m_{1:t})\sim \prod_{k=1}^t\varphi(\eta_k)\textbf{1}_{B_k(\eta_{<k})}\propto \pi_t(\eta_{1:t})$.
\IF{$ESS(w_t^{1:M})<\alpha^\star$}
\STATE
	$Z\leftarrow Z\times \lbrace\frac1M\sum_{i=1}^M w^{i}_t\rbrace$.\\			
	{\bf Resample}: $\eta_t^{\prime m}\sim \sum_{j=1}^M w_t^j\delta_{\eta_t^j}$,  $w_t^m\leftarrow1$.\\
	{\bf Move}: $\eta_t^m\sim K_t(\eta_t^{\prime m},d\eta_t^m) $ where $K_t$ leaves $\pi_t(\eta_{1:t})$ invariant.
\ENDIF
\STATE $\eta^m_{t+1}\sim \varphi(.\vert B_{t+1}(\eta^m_{<t+1}))$,  $w^m_{t+1}\leftarrow w^m_t\times\Phi(B_{t+1}(\eta_{<t+1}))$.
\ENDFOR
\RETURN $Z\times\lbrace\frac1M\sum_{i=1}^M w^i_T\rbrace$ 
\end{algorithmic}
\label{alg:SMC}
\end{algorithm}

An interesting feature of Algorithm \ref{alg:SMC} is that if the integral is simple enough the ESS will never fall under the threshold and the above algorithm breaks down to a GHK simulator. This allows the algorithm to adapt to simple cases at a minimal effort, that of computing the ESS.

Note also that the estimator is still unbiased (\cite{DelMoral1996}) and can therefore be used in more complex schemes such as PMCMC (\cite{Andrieu2010}) or SMC$^2$ (\cite{Chopin2013a}).

\subsection{Move steps}
\label{sec:move}
The moves step will have an important impact on the non-degeneracy of the particle system. We want to construct a Markov chain that moves the particles as far away from their initial position as possible. In addition this step will be the bulk of the added time complexity compared to GHK, so we want to make it as efficient as possible. 

\subsubsection{Gibbs sampler}
The structure of our target \eqref{eq:target}, where the dependence of the Gaussian components lies within the truncation, does not allows a direct application of the Gibbs sampler of \cite{Robert1995}, without a change of variable. In this section, to simplify notations we consider the special case of $\textbf{b}=\infty$. We write the conditional distribution at time $t$ as proportional to $\prod_{i=1}^t \varphi(\eta_i)\ind_{(\Gamma_{<t,<t}\eta_{<t})>b_{<t}}$, where $\Gamma_{<t,<t}$ is the matrix built with the first $t-1$ lines and columns of $\Gamma$. The conditional is given by:
\[
	\eta_i\vert \eta_{-i}\sim\varphi\left(.\Big\vert \bigcap_{j\geq i}^t \left\lbrace \text{sign}(\gamma_{ji})\eta_i\geq \frac1{\vert\gamma_{ji}\vert}\left(a_j-\sum_{k\neq i}\eta_{k}\gamma_{jk}\right)\right\rbrace\right). 
\]
We therefore have to compute those sets for each component up to $t$ and simulate according to a truncated Gaussian. Computing the set can lead to one or two sided truncations depending on the sign of the $\gamma_{ij}$.

The main drawback about having to compute this step each time we resample is its complexity. This operation has time complexity $\mathcal{O}(d^3)$ per time step. This is easily seen as the set in the above equation is just the result of some matrix inversion for a lower triangular system of dimension $d$. This leads to an SMC algorithm that seem to have a prohibitive complexity of $\mathcal{O}(d^4)$, where the GHK simulator had an $\mathcal{O}(d^2)$ complexity. However we have shown  that GHK's variance diverges exponentially quickly on some examples suggesting that this complexity might be acceptable.  In fact examples in high dimension show that even at constant computational cost the algorithm is able to over-perform GHK (see Section \ref{sec:NumAna}).

\subsubsection{Hamiltonian Monte Carlo}
An alternative to Gibbs sampler is to use Hamiltoninan Monte Carlo (HMC) (see \citep{Neal2010} for a survey), and the idea of \cite{Pakman2012} for truncated Gaussians.

HMC is based on interpreting the variables of interest as the position of a particle with potential the opposite of the log target and by simulating the momentum as a Gaussian with given mass matrix. The proposal of the Metropolis-Hastings is then constructed by applying the equations of motion up to a time horizon $T_{HMC}$ to the problem. This leads to an efficient algorithm that makes use of the gradient of the target to explore its support. We refer the reader to \cite{Neal2010} for more details on the algorithm and describe the approach proposed by \cite{Pakman2012} to adapt the algorithm to truncated Gaussians. 

Based on the fact that the log density of a Gaussian random variable is a quadratic form, the movement equation can be dealt with explicitly. The scheme is written as an exact HMC (i.e. not resulting in numerical integration). Remains then to deal with the truncation. \cite{Pakman2012} show that they can be treated as ``walls'' for the given particle, a reflection principle can be applied for any particle hitting the constraint during the algorithm. In particular we must find the time at which occurs the first ``hit''. In our experiment the time horizon $T_{HMC}$ is set to a uniform  draw on $\left[0,\pi\right]$ as suggested in \cite{Neal2010}. The average value $\pi/2$ is advocated by \cite{Pakman2012}.

The computation of the first hitting time dominates the cost of the algorithm. This is particularly true when the truncation are small as the number of hitting times will be high. Figure \ref{fig:HMC} in appendix \ref{app:HMC} shows a comparison of the SMC algorithm with the Gibbs sampler (grey) and exact HMC (white). Although this Markov chain algorithm seems to perform very well for a wide range of problems and has a neat formalism, we find that it does not outperform Gibbs sampling when used as a move. The specificity of the move step in SMC is that the particles are already distributed according to $\eqref{eq:target}$, therefore the move need not propagate each particle across all the support. In particular the strength of HMC in quickly exploring the target might be less useful in this context.

\subsubsection{Overrelaxation}
Overrelaxation for Gaussian random variables was proposed by \cite{Adler1981} as a way of improving Gibbs sampling for a distribution with Gaussian conditionals.  

For each component the proposal is $\eta^\prime_i\vert \eta_{-i}\sim \mathcal{N}(\mu_i+\alpha(\eta_i-\mu_i),\sigma^2_i (1-\alpha^2))$ for $0\leq\alpha\leq 1$, and with $\mu_i$ and $\sigma^2_i$ the expectation and variance of $\eta_i\vert \eta_{-i}$. The case $\alpha=0$ is the classical Gibbs sampler, the case $\alpha=1$ is a special case of random walk Metropolis-Hasting proposal. One can check that if $\eta_i\sim\mathcal{N}(\mu_i,\sigma^2_i)$ then $\eta^\prime_i$ has the correct distribution. 

Given a particle $\eta$, we propose a new one according to:
$$
\eta^\prime\vert \eta\sim \mathcal{N}(\alpha \eta, (1-\alpha^2)I)
$$
Setting aside the constraint for a moment the invariant distribution of such kernel is an independent (0,1)-Gaussian. If we add an acceptation  step such that we accept if it satisfies the constraint at time $t$, the Markov kernel leaves the current distribution invariant \eqref{eq:target}. 

We find that the fact that overrelaxation is close to a Metropolis adjusted Langevin algorithm (MALA) helps to calibrate the algorithm, $\log \pi (\eta)=-\frac12\eta^T\eta$, hence the proposal in MALA is $\mathcal{N}((1-\frac{\varepsilon}{2})\eta,\varepsilon^2)$. From \cite{Roberts1998} we have that $\varepsilon$ should be $\mathcal{O}(d^{-\frac13})$. To calibrate the algorithm we propose to match the two drifts. We find that $\alpha=\mathcal{O}(1-0.5d^{-\frac13})$, the constant should then be close from a problem to an other because locally we are always in the case of independent Gaussians (locally the constraints have less impact). We find that in our case taking $\alpha=0.004\times (1-d^{-1/3})$ gives the expected behavior and acceptance ratio.

\subsubsection{Repeating the move step} 
\label{subsub:repeat}
\cite{Dubarry2011} have shown, for particle filters, that applying some Metropolis Hastings kernel targeting the filtering distribution on the particles leads to a close to optimal variance (the variance is the same as one coming from an $iid$ sample). This convergence results happens after $\mathcal{O}(\log M)$ iteration of the Markov kernel. These results suggest repeating the move step after each resampling step until some criterion of convergence is satisfied.

We compute the sum of absolute distances that the particles have moved after each step (a similar metric was used in \cite{Schaeefer2013} for the discrete case). We repeat the move until this scalar value stabilizes. The stabilization of the total metric should be associated with the cancellation of the dependence between the particles (leading to a close to independent system). 

\subsubsection{Block sampling}
To diversify the particle system after each resampling we have relied until now on invariant kernels targeting the current distribution $\pi_t$. An alternative to this approach is given by \cite{Doucet2006}, where importance sampling is done on the space of $\eta_{t-L+1:t+1}$ with a given number $L$ of previous time steps. This limits the behavior of the particles all stemming from one path after a few iterations. We briefly describe the idea in the following. 

Suppose at time $t-1$ we have a weighted set of particles such that $(w_{t-1},\eta_{1:t-1})\sim \pi_{t-1}(\eta_{1:t-1})$; instead of proposing a particle $\eta^\prime_t$, propose a block of size $L$, $\eta^\prime_{t-L+1:t}\sim q(.\vert \eta_{1:t-1})$, and discard the particles $\eta_{t-L+1:t-1}$. The distribution of the resulting system is intractable because of the marginalization. However \cite{Doucet2006} note that importance sampling is still possible on the extended set of particles $(\eta_{1:t-1},\eta^\prime_{t-L+1:t})$ by introducing some auxiliary distribution $\lambda_t(\eta_{t-L+1:t-1}\vert \eta^\prime_{1:t})$. This leads to the correct marginal whatever $\lambda_t$ and the algorithm has the following incremental weights:
\[
\frac{\pi_t(\eta_{1:t-L},\eta^\prime_{t-L+1:t})\lambda_t(\eta_{t-L+1:t-1}\vert \eta^\prime_{t-L+1:t},\eta_{1:t-1})}{\pi_{t-1}(\eta_{1:t-1})q(\eta^\prime_{t-L+1:t}\vert \eta_{1:t-1})}.
\]

The authors show that the optimal proposal and resulting weights are given by:
\[
q^{opt}(\eta^\prime_{t-L+1:t}\vert \eta_{1:t-1})=\pi_t(\eta^\prime_{t-L+1:t}\vert \eta_{1:t-L}),
\]
$$
w_t=w_{t-1} \frac{\pi_t(\eta^\prime_{1:t})}{\pi_{t-L}(\eta_{1:t-L})}.
$$
In our case the optimal proposal can then be shown to be:
$$
q_t^{opt}(\eta_{t-L+1:t}^\prime\vert \eta_{1:t-L})=\frac{\prod_{i=1}^t\ind_{B_i(\eta_{<i})} \prod_{i=t-L+1}^t \varphi(\eta_i)}{\int\prod_{i=1}^t\ind_{B_i(\eta_{<i})} \prod_{i=t-L+1}^t \varphi(\eta_i)d\eta_{t-L+1:t}}.
$$
Notice that this is the density of a truncated Gaussian distribution, yielding a weight depending on an orthant probability (denominator). In most cases this is not available and in our particular case it is the quantity of interest. We can however compute explicitly this integral for $L=1$ and $L=2$. The former is the usual case (block of size one). The case $L=2$ did not bring any improvement in terms of variance in all our simulation. We concentrated on the extension to blocks of higher dimension.

In this case we have to resort to approximations of the proposal. The first idea would be to approximate it by a Gaussian using expectation propagation \citep{Minka2001}. However this approach did not perform better than the use of Gibbs sampler mentioned earlier. Another approach to approximate the distribution is to consider the Gibbs sampler on a block of size $L$ with the GHK proposal.

\subsubsection{Partial conclusion}
We have shown that the proposed Gibbs sampler outperforms HMC. Concerning block sampling the different approaches were tested on several dimensions only to find that the best performing approach was to use partial Gibbs sampling, i.e. a Gibbs sampler on a block. In the numerical tests we provide in Section \ref{sec:NumAna} we show only the latter. 

In our simulations we propose to repeat each kernels as was explained in Section \ref{subsub:repeat}. We propose to test the Gibbs sampler and the overrelaxed random walk.

In addition we have studied other kernels based on the geometry of the problem; in particular, one can draw random walks on the line between the current particle and the basic solution of our constraint. Those approach did not however outperform the proposals discussed above.  

\section{Extentions}
\label{sec:Extensions}
\subsection{Student Orthant}
\label{sec:Student}

We can easily extend our approach to the computation of orthant probabilities for other distributions, in particular for mixtures of Gaussians, that is probabilities that can be written as:
\begin{equation}
\label{augm}
\int f_U(u)\int f_{H|U}(\eta\vert u)\indicator{b>\eta>a}d\eta du,
\end{equation} 
where $f_{H\vert U}$ is a Gaussian. Several distributions can be created as such. For instance, the Student distribution  where the variance is marginally distributed as an inverse-$\chi^2$. Hence the distribution: 
\[
	f_{H|U}(\eta|u)\indicator{b>\eta>a}=\prod_{i=1}^n\varphi(\eta_i)\indicator{B_i^u(\eta_{<i})},
\]
where $B^u(\eta_{<i})$ is $B_i(\eta_{<i})$ where we multiply $a$ by $\frac u\nu$ and $f_U(u)=\chi^2_\nu(u)$. They are an interesting application to those algorithms because they come at a minimal additional cost and are of use in multiple comparison \citep{Bretz2001}.

Another example is the logistic distribution where $f_U(u)$ is some transformation of a Kolmogorov-Smirnov distribution (see \cite{Holmes2006}). This could be used to perform Bayesian inference on multinomial logistic regression.  

To deal with this integral we can extend the space on which the SMC is carried out at time $t$. Hence the move step is performed on the extended space $f_U(u)f_{H\vert U}(\eta\vert u)$. In our Student example it amounts to taking as a target distribution 
\[
	\pi_n(\eta_{1:n},u)\propto \prod_{i=1}^n\varphi(\eta_i)\indicator{B_i^u(\eta_{<i})}\chi^2_\nu(u).
\]
The normalizing constant that the SMC algorithm approximates is
\[
	Z_n=\int\prod_{i=1}^n\varphi(\eta_i)\indicator{B_i^u(\eta_{<i})}\chi^2_\nu(u)\dd\eta_{1:n}\dd u.
\]
At each move step we therefore move the particles using a Metropolis-Hastings algorithm targeting $p(u\vert \eta_{1:n})$ and perform the remaining Gibbs sampler updates conditionally on $U$. This additional step allows for further mixing. Benefits from this step are already found in relatively low dimension as shown in Section \ref{sec:NumAna}. 
\subsection{SMC as a truncated distribution sampler}

A natural extension is to use Alg. \ref{alg:SMC} to compute other integrals with respect to truncated Gaussians. At time $t$ the output of the algorithm is a weighted sample $(w^i_t,\eta^i_{1:t})_{i\in \left[1,M\right]}$ approximating $\pi_t(\eta_{1:t})\propto\prod_{i=1}^t\varphi(\eta_i)\ind_{B(\eta_{<i})}$. Hence any integral of the form $\mathbb{E}_{\pi_t}\left(h(\eta)\right)$, where expectation is taken with respect to $\pi_t$, can be approximated by $\sum_{i=1}^M\frac{w^i_t }{\sum^M_{j=1}w_t^j}h(\eta_{1:t}^i)$. The same argument goes for the truncated Student.  

We test the idea for computing the expectation of truncated multivariate Student. We use a Gibbs sampler as a benchmark based on \cite{Robert1995}'s sampler by adding a MH step to deal with $u$ (see previous section). The Gibbs update is done after a change of variable that leaves the truncations independent. This can be shown to be more efficient. We allocate $100$ times more computational time to the Gibbs sampler than the SMC. 

In Figure \ref{fig:mean} we see that after thinning one out of $1000$ points the ACF and trace plots point to bad exploration of the target's support. This behavior shows that the convergence is too slow for the algorithm to be of practical use. On the other hand the SMC is still stable as is shown in the next section.

In addition of outperforming the Gibbs sampler for fairly moderate dimension, the SMC algorithm was found to be stable for approximating the expectation in dimensions up to $100$.
\begin{figure}[H]
\begin{center}
\begin{tabular}{cc}
\subfloat[Trace plot]{\includegraphics[scale=0.25]{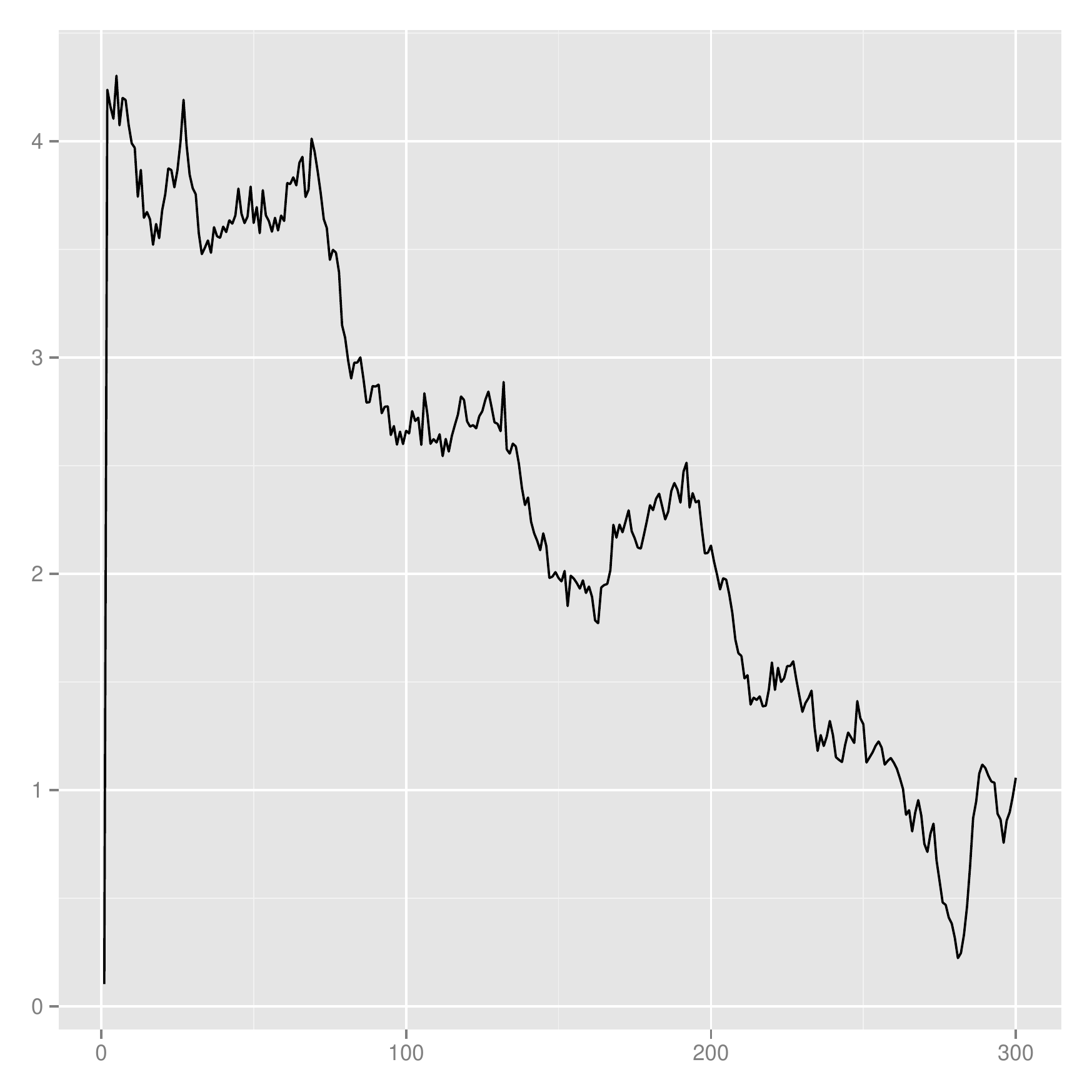}}&
\subfloat[ACF]{\includegraphics[scale=0.25]{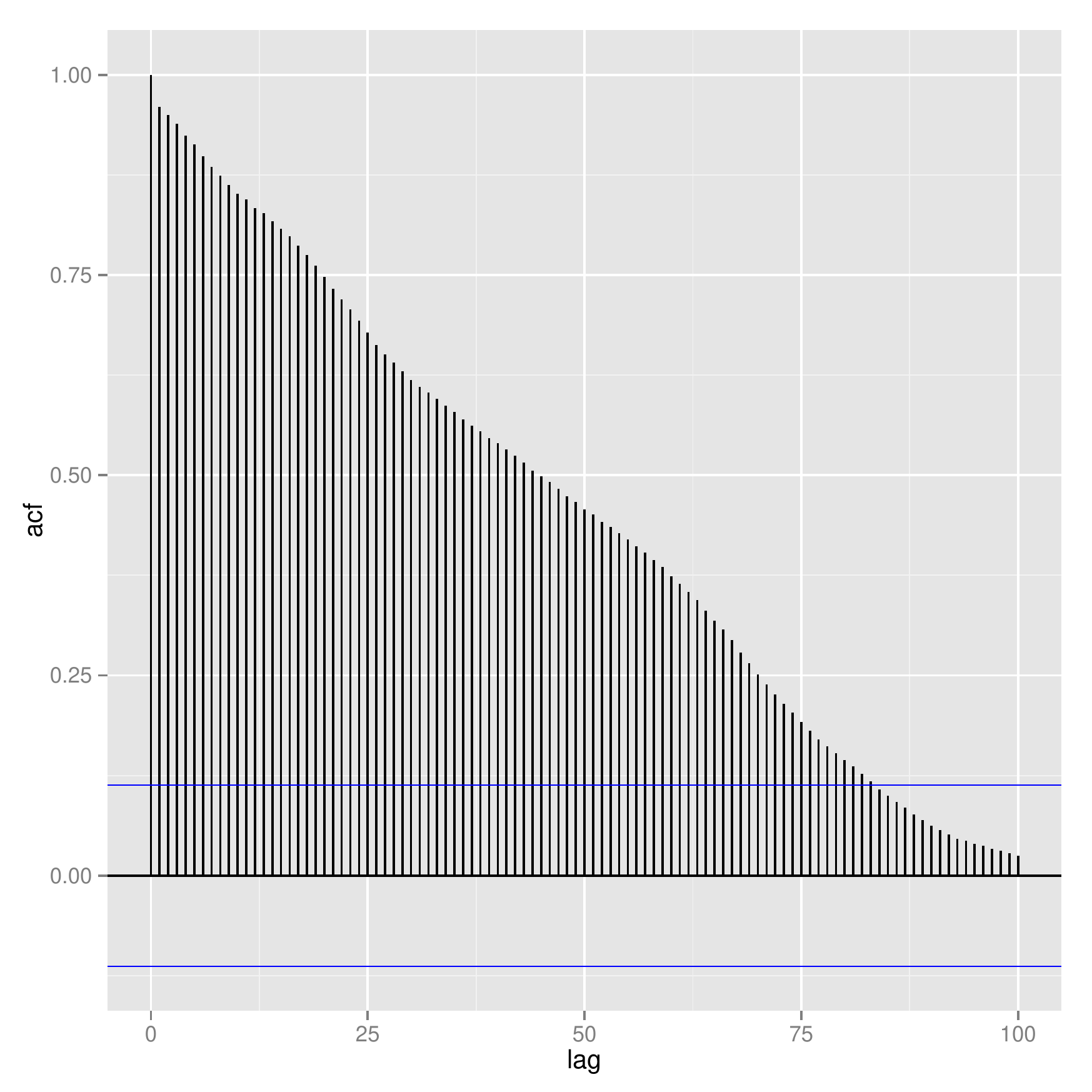}}\\
\subfloat[Trace plot]{\includegraphics[scale=0.25]{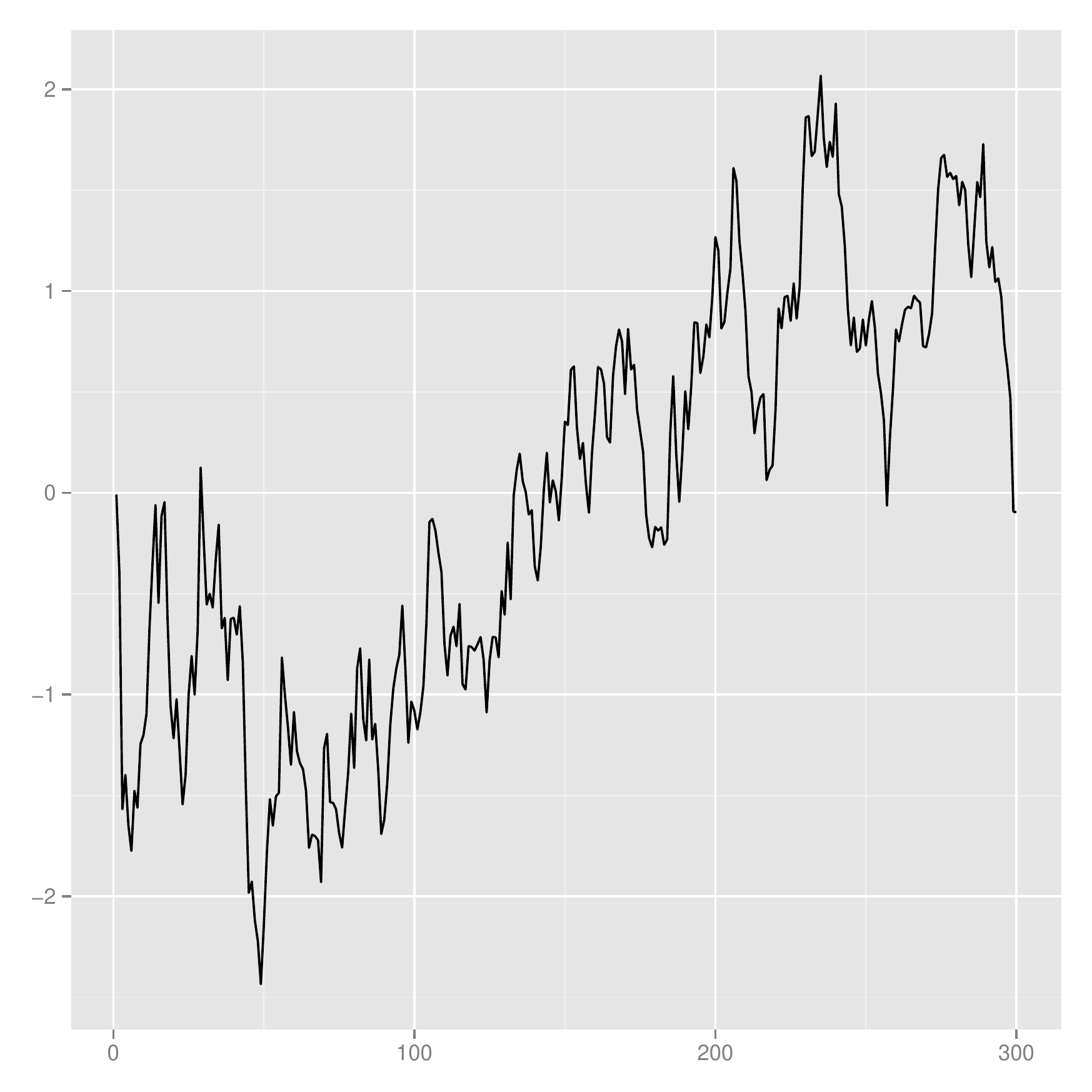}}&
\subfloat[ACF]{\includegraphics[scale=0.25]{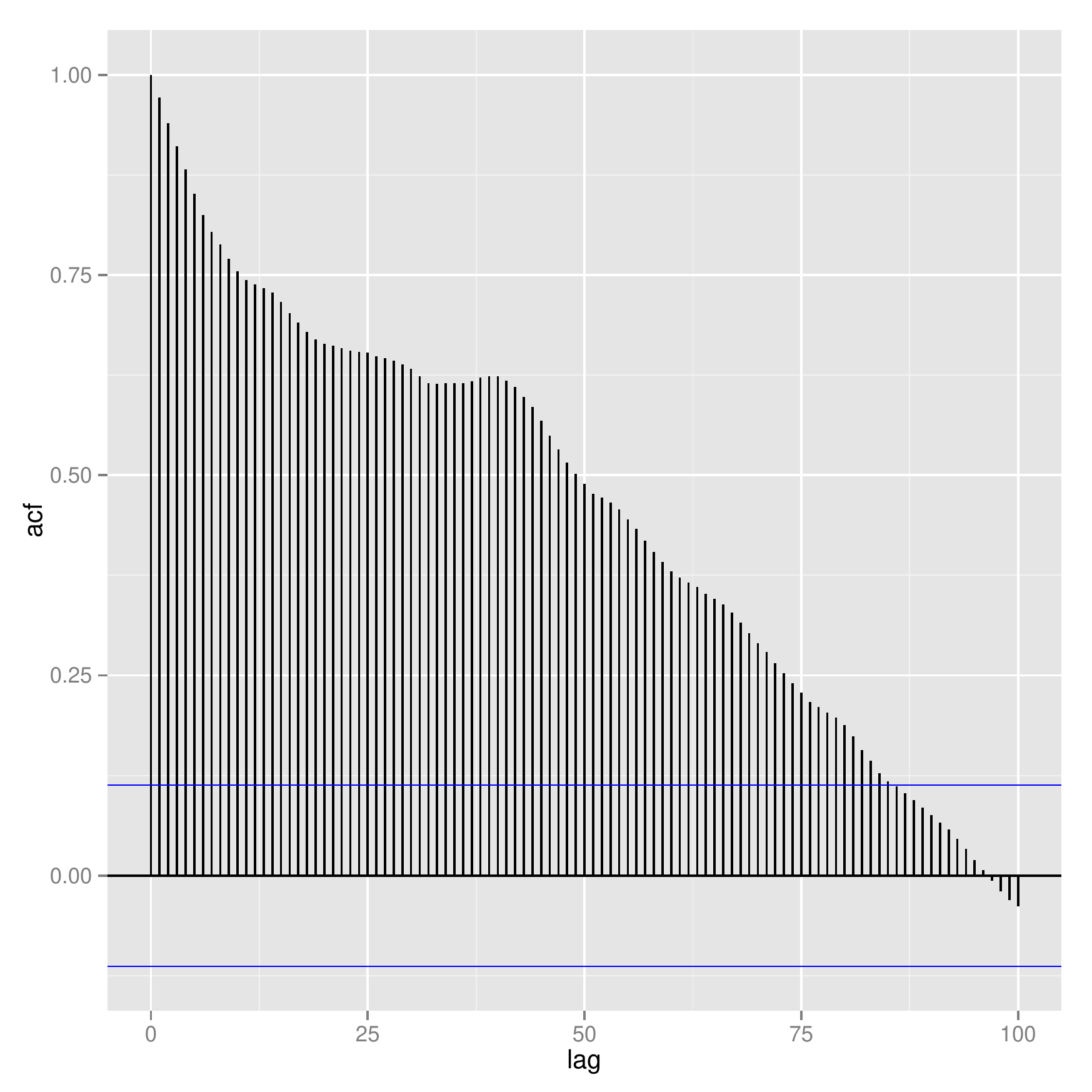}}\\
\end{tabular}
\caption{Truncated Student sampling after thinning one out of every 1000, using a Gibbs sampler}
\label{fig:mean}
\begin{minipage}{15cm}
\footnotesize{The data are generated as explained in the ``Numerical results section'', for dimension $50$. The left panels are  trace plots for two components. The right panels are the ACFs. Both are shown after a thinning of 1/1000. Both show a slow convergence, whereas we observe that SMC is stable.}
\vspace*{3mm}
\hrule
\end{minipage}
\end{center}
\end{figure}


\section{Numerical results}
\label{sec:NumAna}
\subsection{Covariance simulation, tunning parameters}
To build the covariance matrices we propose to use draws from a Cauchy distribution. We start by sampling a matrix $X$ and a vector $\textbf{a}$ from an independent Cauchy distribution $X_{ij}\sim\mathcal{C}(0,0.01) \text{  and }a_i\sim\mathcal{C}(0,0.01)$,
then construct the covariance matrix as $\Sigma=X^tX$ and the truncation as $\textbf{a}=(a_i)$. Because of the heavy tails the resulting correlation matrix (figure \ref{fig:heat}) has many close to zero entries and some high correlations. The truncations also have some very high levels (figure \ref{fig:trunc}).

We find that this approach leads to more challenging covariances than those built by sampling the spectrum of the covariance matrix as proposed for instance in \cite{christen2012optimal}.

\begin{figure}[H]
\begin{center}
\begin{tabular}{cc}
\subfloat[Correlations]{\includegraphics[scale=0.3]{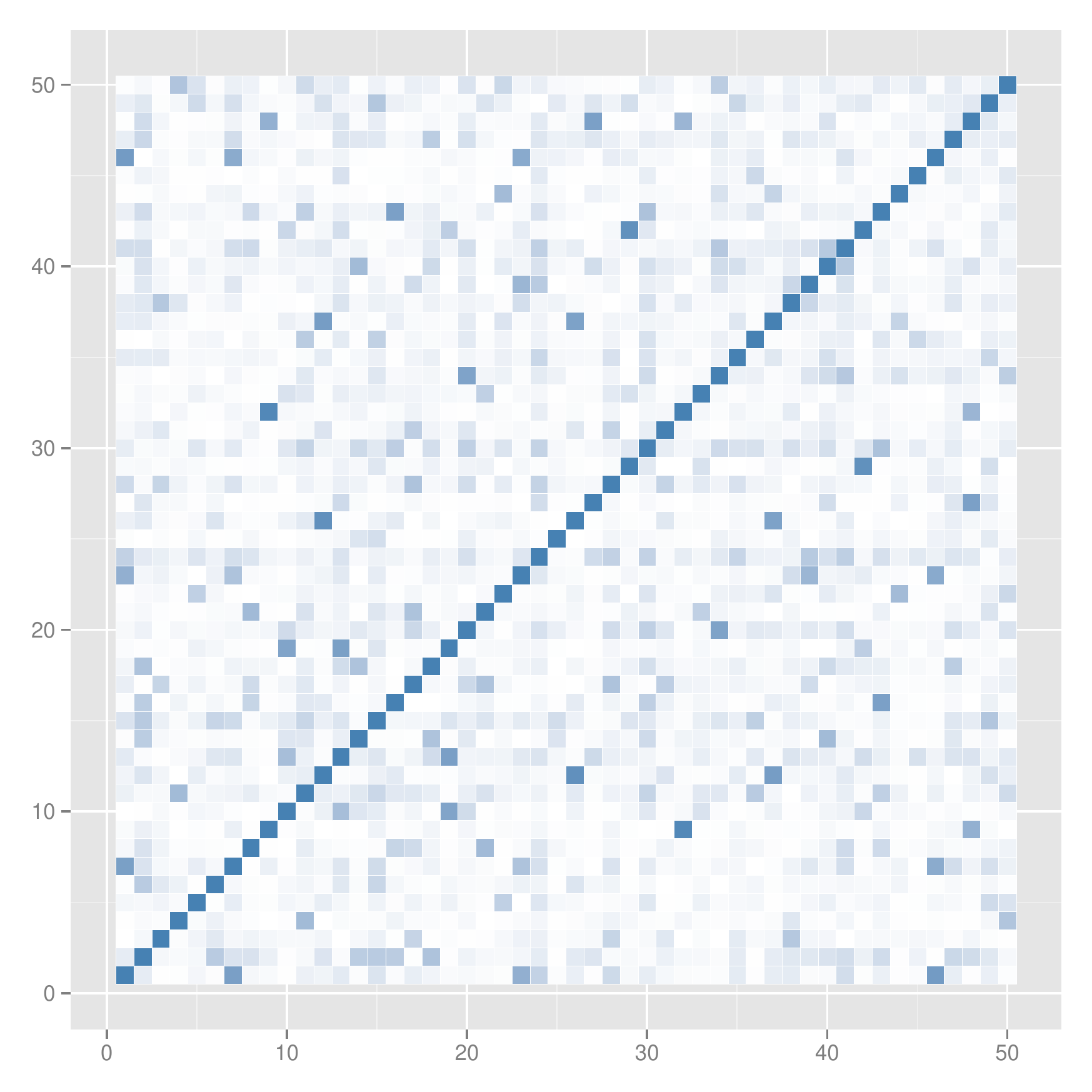}\label{fig:heat}} &
\subfloat[Truncation]{\includegraphics[scale=0.3]{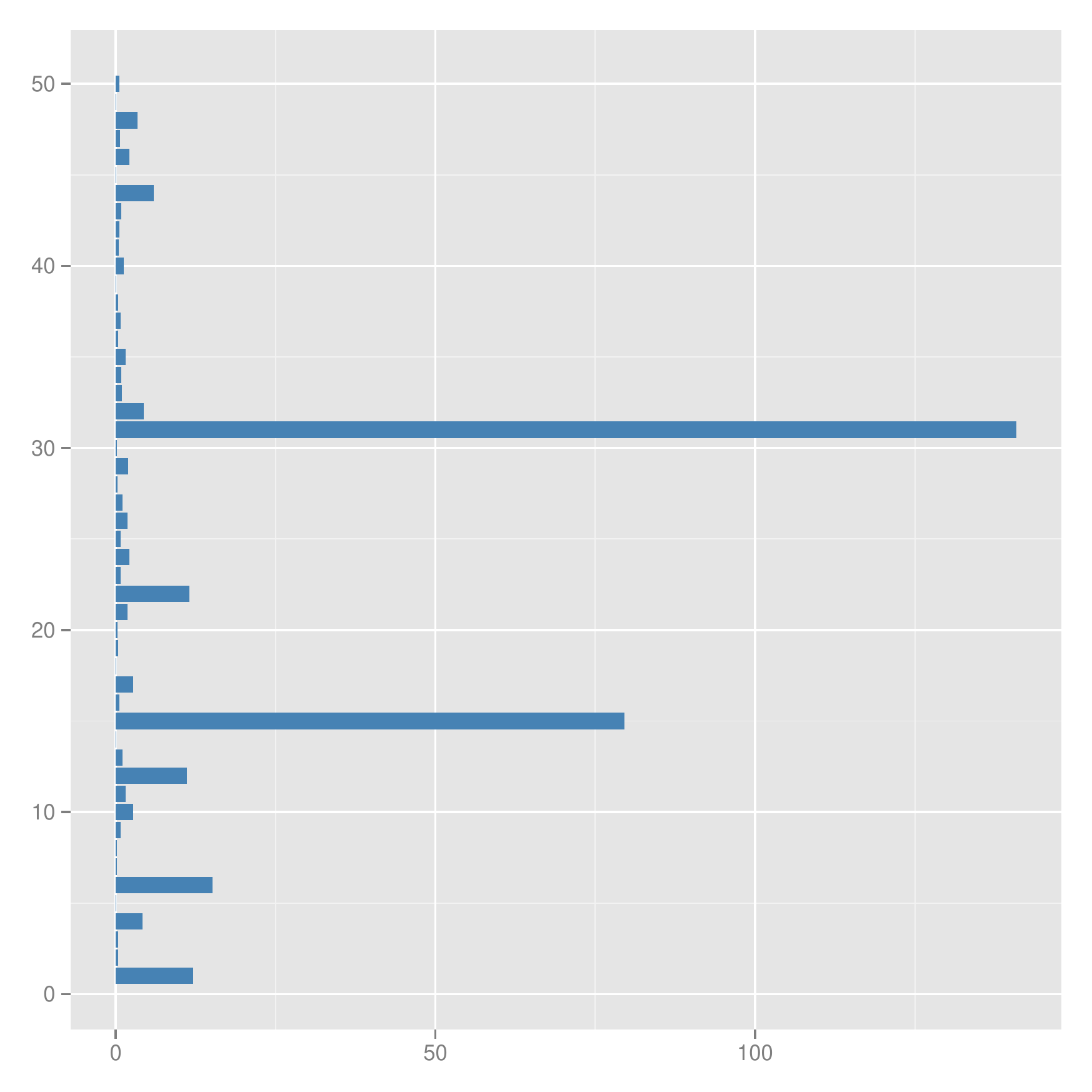} \label{fig:trunc}}
\end{tabular}
\caption{Generated correlations}
\label{fig:hmc}
\begin{minipage}{15cm}
\footnotesize{Covariance matrices generated from random samples with heavy tails. The various dimension are simulated with the same algorithm and same seed, such that the small ones are subsets of the others. The left panel shows a heatmap of the correlation, the right panel the left trucation of the integral.}
\vspace*{3mm}
\hrule
\end{minipage}
\end{center} 
\end{figure}

The tuning parameters of the algorithm are the threshold that tunes the number of steps of the MCMC kernel, and the targeted ESS under which we resample, $\alpha^\star$. The former is set to some small value (0.01) and has not much influence. The latter gives us a trade off between variance and computational cost. In our example we have found that $0.5M$ allows good approximation, however this value should be increased with the difficulty of the problem. 
\subsection{GHK for moderate dimensions}

All our results are shown at constant computational cost: we repeat the algorithm in a first time to get their execution time, and  we then scale them accordingly. In the above example (dimension 40) for instance the number of draws associated with the GHK algorithm is $1,065,399$. The number of particles of the SMC sampler with Gibbs Markov transition is $5217$.

\begin{figure}[H]
\begin{center}
\begin{tabular}{cc}
\subfloat[Dimension 20]{\includegraphics[scale=0.3]{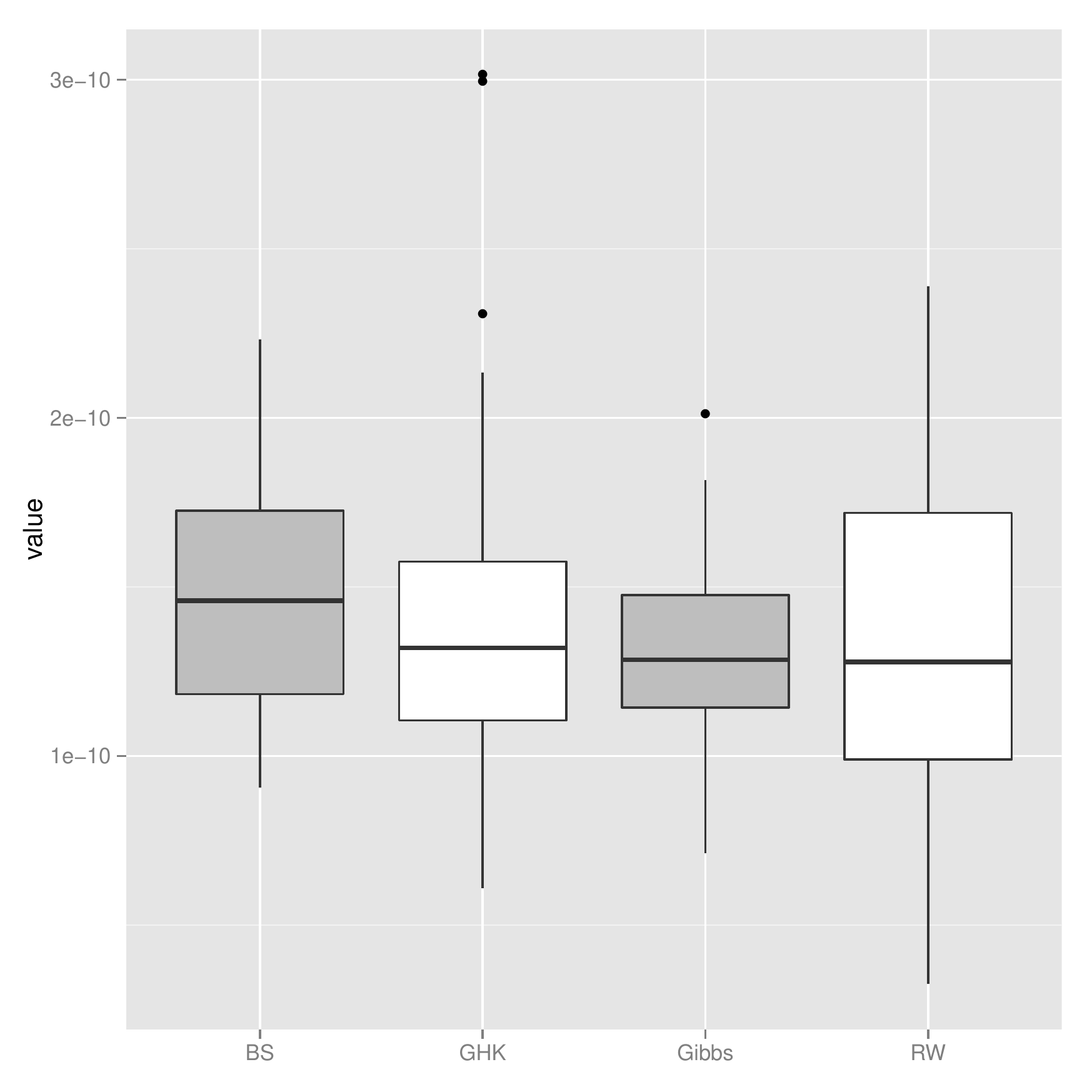}} &
\subfloat[Dimension 30]{\includegraphics[scale=0.3]{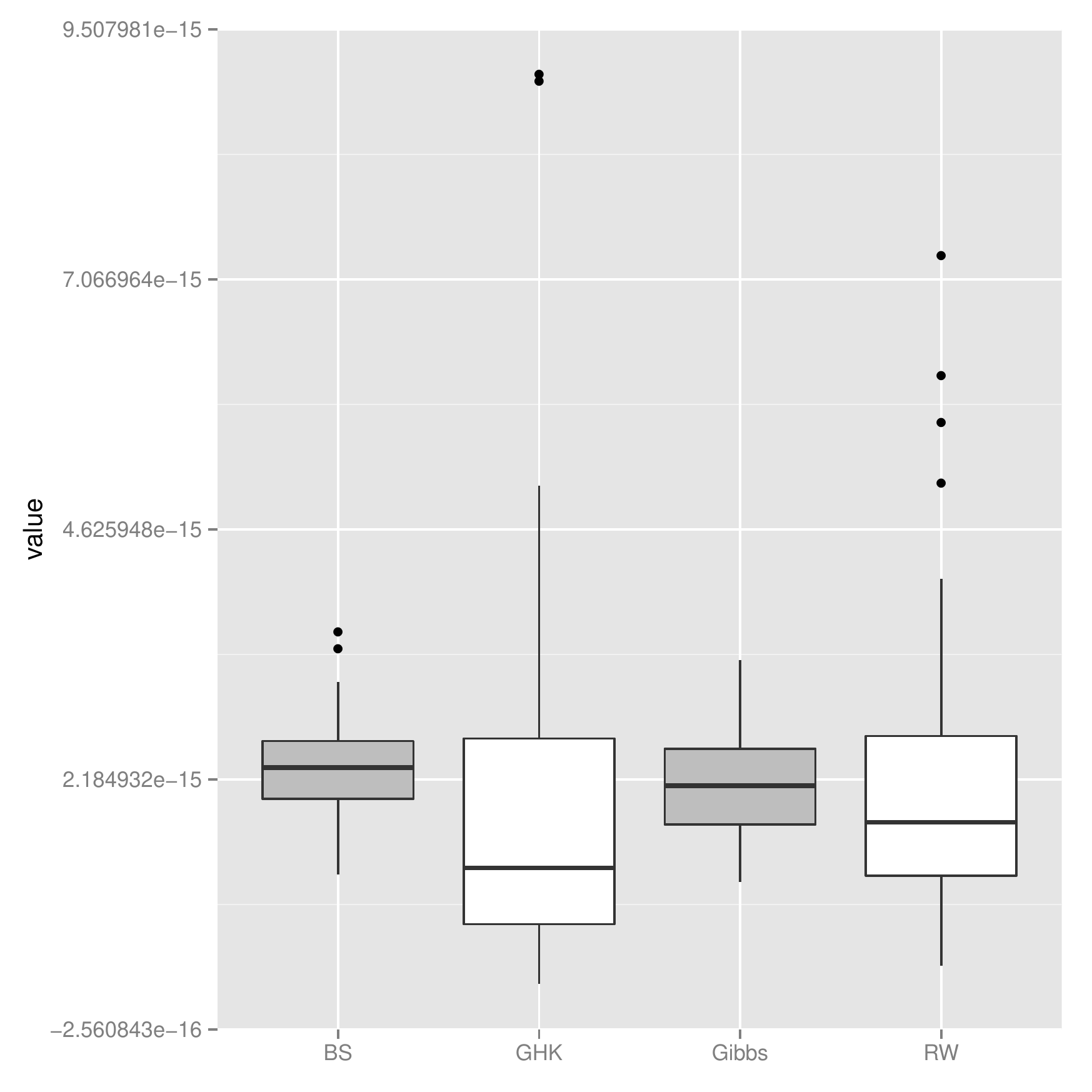} }\\
\subfloat[Dimension 40]{\includegraphics[scale=0.3]{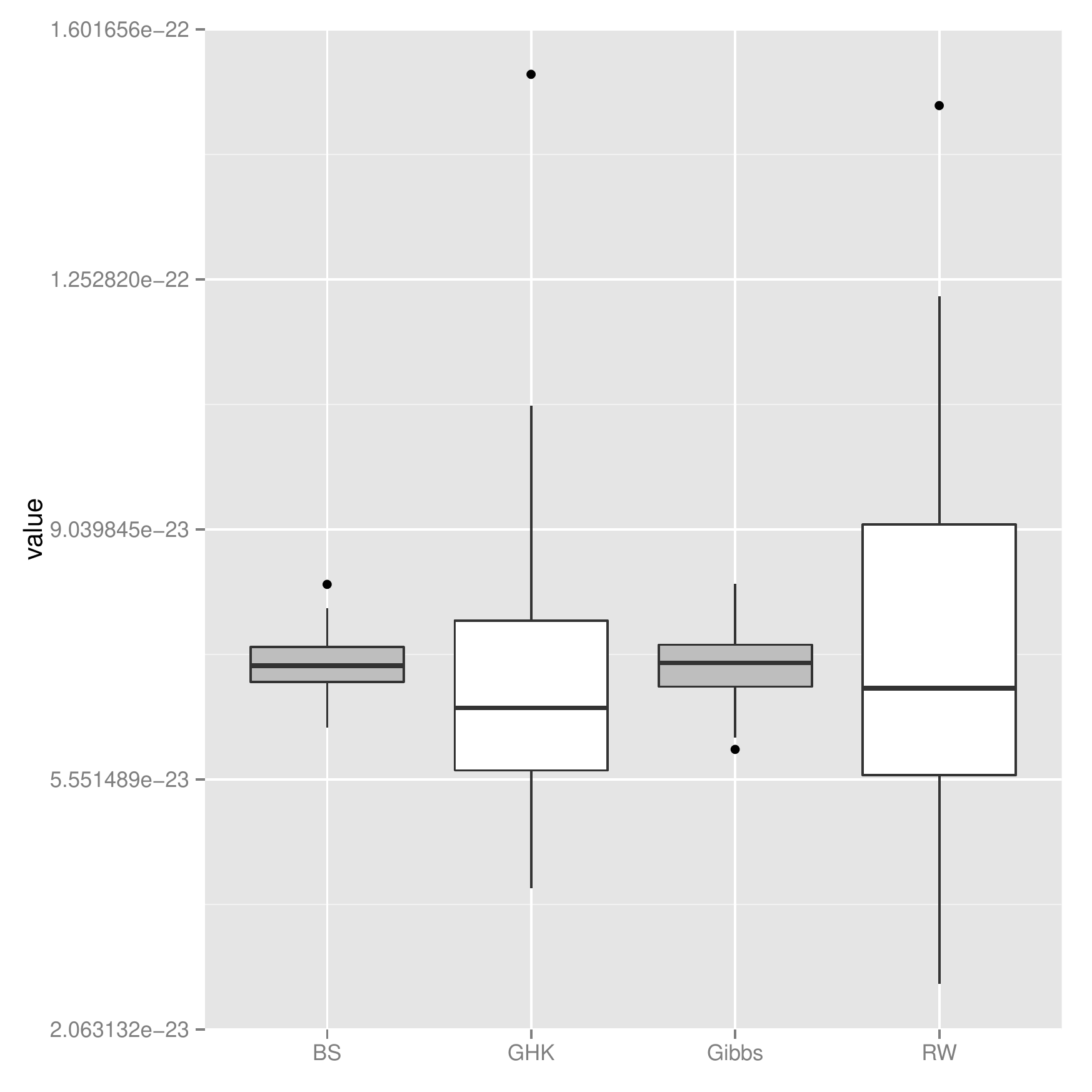} } &
\subfloat[Dimension 50]{\includegraphics[scale=0.3]{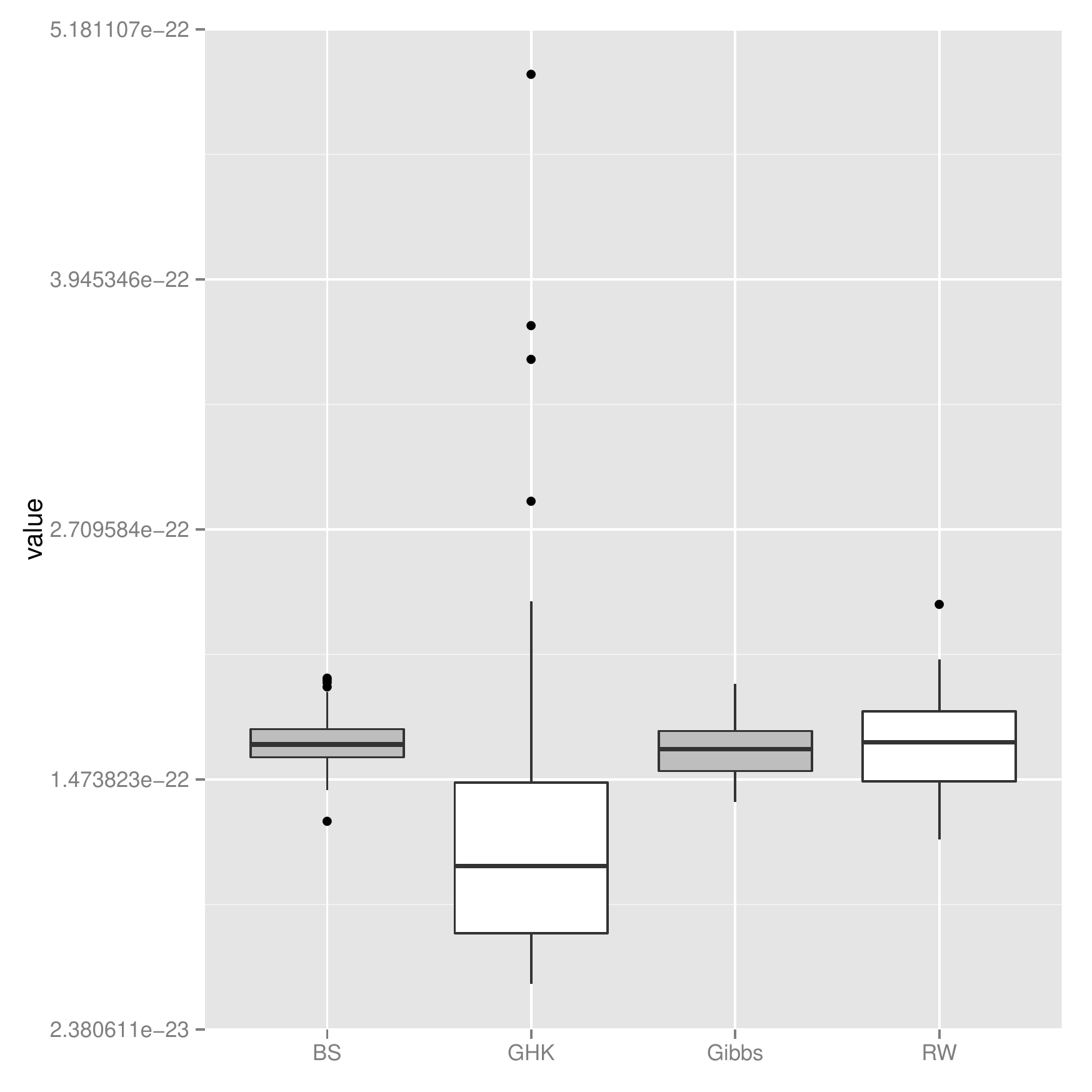} }
\end{tabular}
\caption{Estimates of Orthant probabilities for GHK compared to SMC with different moves.}
\label{fig:comp}
\begin{minipage}{15cm}
	\footnotesize{The various dimensions are simulated as described in Section \ref{sec:NumAna}. Different moves are tested inside SMC. As we have already discussed GHK leads to a higher variance and outliers. Block sampling (BS) and Gibbs sampling (Gibbs) seem to outperform the overrelaxed random walk (RW). However all three stay stable until dimension 50.}
\vspace*{3mm}
\hrule
\end{minipage}
\end{center}
\end{figure}

We find that in moderate dimension ($\sim 50$) the GHK simulator breaks down in attempting to compute the probability of the orthant generated by our simulation scheme. This is all the more problematic as it gives an answer, and there is no way of checking its departure from the true value.

Another interesting aspect is the fact that the block sampling algorithm performs well in those dimensions. It is quicker than to move the particles in every dimension as is done with MCMC. The truncations that lead to a drop of probability are close together because of the ordering, hence once the difficult dimensions have been ``absorbed'' it is less and less paramount to visit the past truncations.

\subsection{High dimension orthant probabilities}

In dimensions higher than $p=70$, the covariance we simulate lead to integrals that cannot be treated with the GHK algorithm. In our simulations GHK always returned NaN values due to the low values of the weights. For the SMC an indicator of the good behavior of the algorithm can be seen in either its reproducibility and the fact that we do not encounter asymmetry (see Remark \ref{rmk:lognorm}) as for the GHK in the first two Sections. Furthermore the ESS does not fall very low along the particles' draw (Figure \ref{fig:esshd}). 

\begin{figure}[H]
\begin{center}
\begin{tabular}{lll}
\subfloat[Dimension 130]{\includegraphics[scale=0.27]{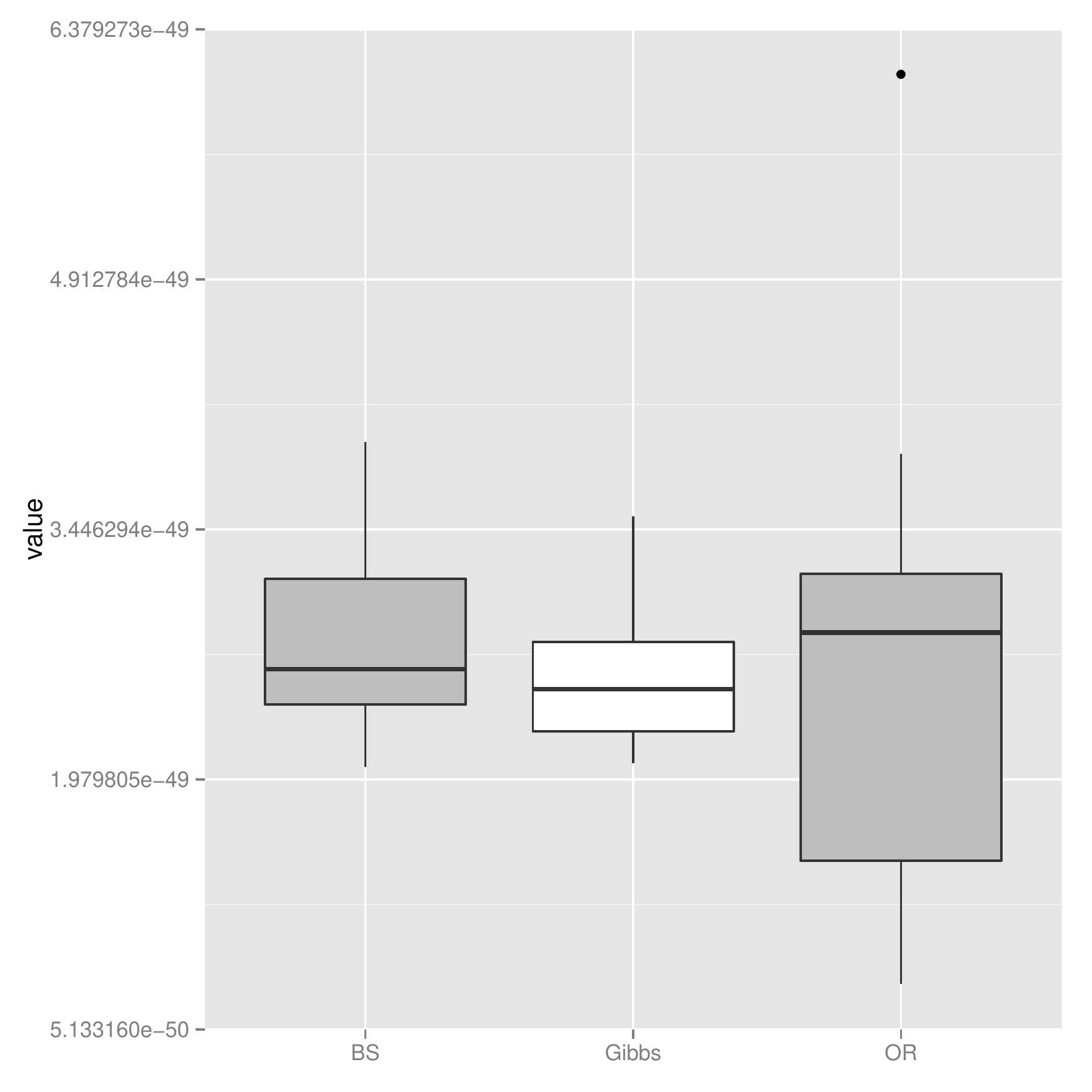}} &
\subfloat[Dimension 180]{\includegraphics[scale=0.27]{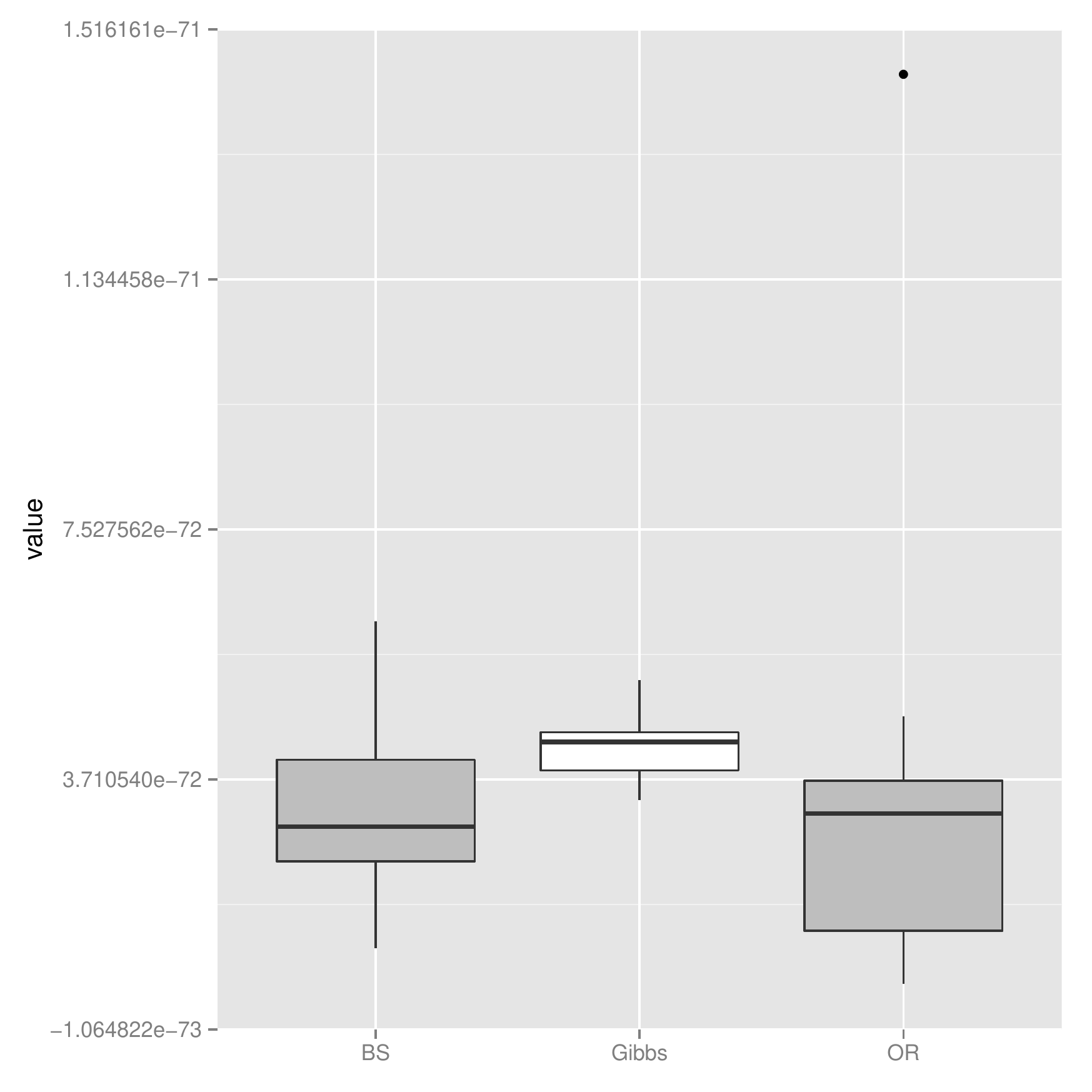}}
\subfloat[Dimension 130]{\includegraphics[scale=0.27]{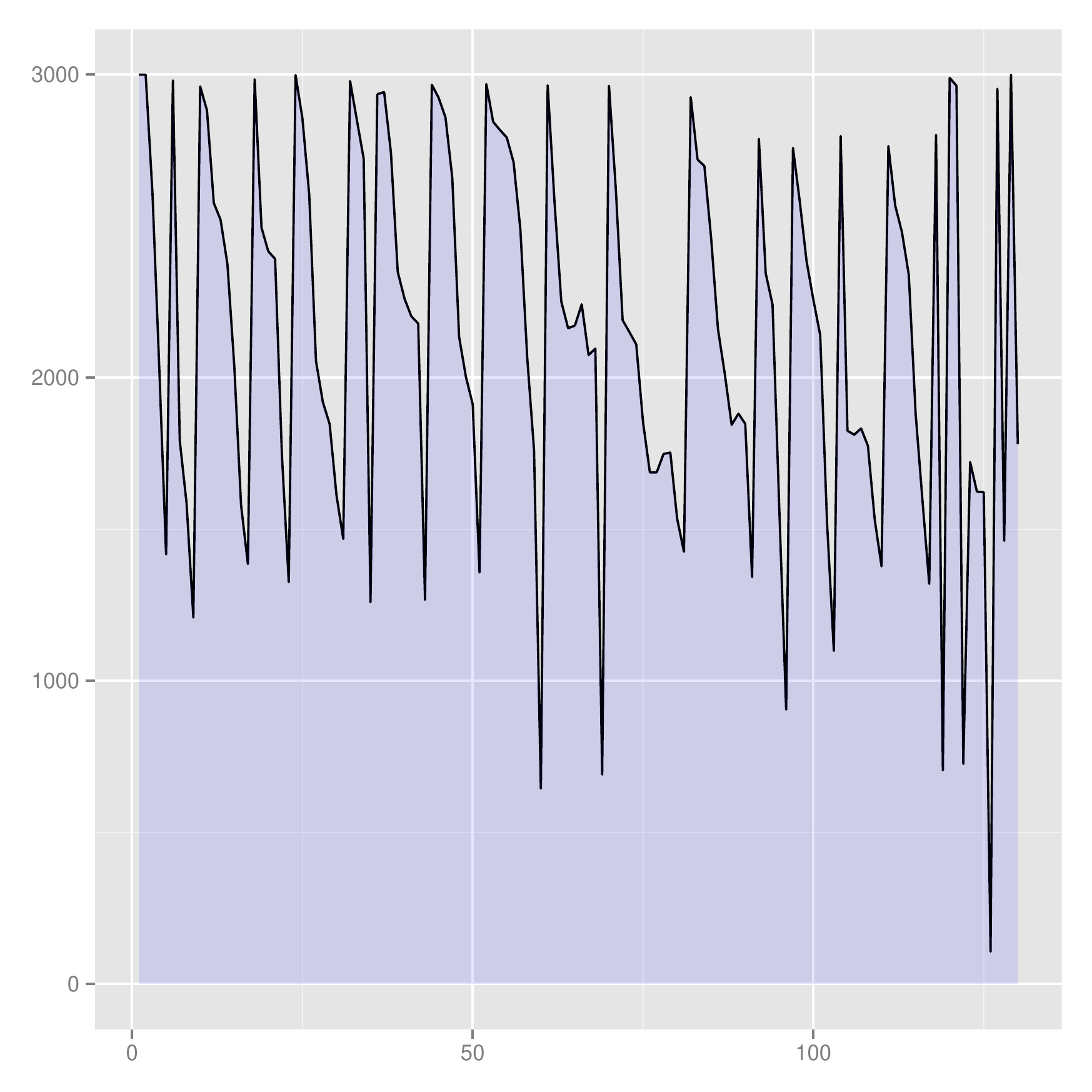}\label{fig:esshd}} 
\end{tabular}
\caption{Estimates of Orthant probabilities $p=130$ and $p=180$}
\label{fig:highdim}
\begin{minipage}{15cm}
	\footnotesize{The various dimensions are simulated as described in Section \ref{sec:NumAna}. Different moves are tested inside SMC. As we have already discussed GHK leads to a higher variance and outliers. Gibbs sampling (Gibbs) seem to outperform the overrelaxed random walk (OR) and Block sampling (BS). However all three stay stable until dimension 180. The ESS for the Gibbs sampler is shown in panel $c$ for a threshold of $0.5M$ and $M=3000$. Despite some sudden drops it seems to be stable.}
\vspace*{3mm}
\hrule
\end{minipage}
\end{center}
\end{figure}

For those dimensions the Gibbs sampler performs best in terms of variance. However if one's goal is a fast algorithm, at the cost of higher variance the overrelaxation might be preferable at some point as the dimension of the target increases. The latter as a complexity smaller of one degree such that at constant computational cost it will have more and more particles allocated to it.

\subsection{Student orthant probabilities}

We use the same schemes as before to construct the covariance matrix and fix a degree of freedom of $3$ in our experiments. As before we show an improvement as compared to previous algorithms. This improvements appears also for moderate dimensions. It seems that there is an important gain in considering the extended target. 

As for the Gaussian case we find that the output of GHK is heavily skewed. It seems that it is not the case for our algorithm.

\begin{figure}[H]
\begin{center}
\begin{tabular}{cc}
\subfloat[Dimension 30]{\includegraphics[scale=0.3]{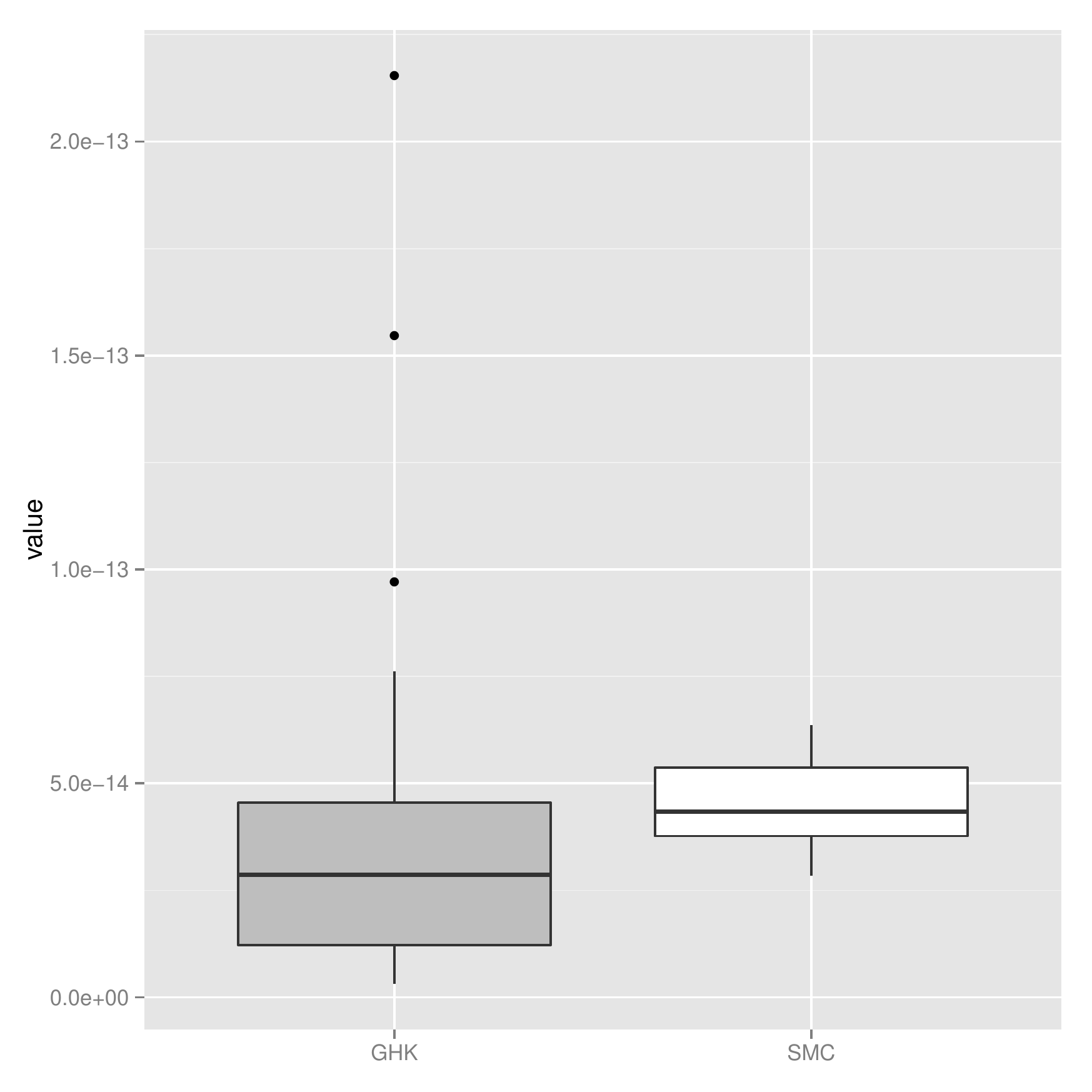}} &
\subfloat[Dimension 50]{\includegraphics[scale=0.3]{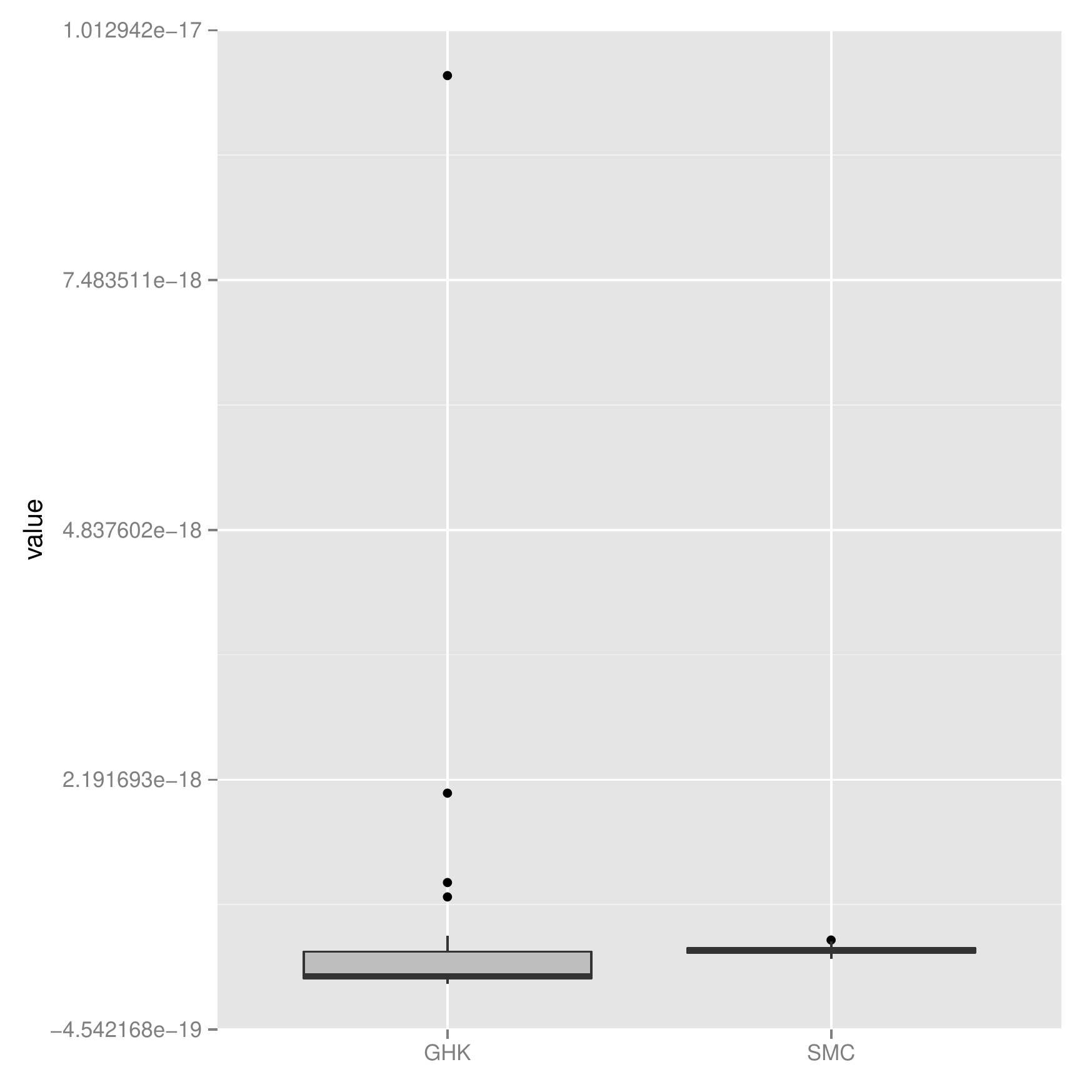} }
\end{tabular}
\caption{Estimates of Student orthant probabilities SMC vs GHK}
\label{fig:Student}
\begin{minipage}{15cm}
\footnotesize{Covariance matrices generated from random samples with heavy tails. We find that the SMC outperforms the GHK. As for the previous cases the GHK leads to some outliers.}
\vspace*{3mm}
\hrule
\end{minipage}
\end{center}
\end{figure}
\subsection{Application to random utility models}
Random utility models are an important area of research in Economics to model choice data \citep{Train2009}. Consider an agent $i$ confronted to $J$ alternatives each giving utility $Y_{ij}^\star$ $\forall j \in \lbrace 1,\cdots, J\rbrace$ modeled by $Y^\star_{ij}=X_i\beta+u_{ij}$ with $u_{ij}$ a Gaussian noise. Individual $i$ chooses alternative $j$ if $\lbrace \forall k\neq j \quad Y_{ij}^\star>Y^\star_{ik}\rbrace$. The likelihood is the probability of this set integrated over the unobserved alternatives.
Hence the likelihood is given by:
\begin{align*}
L(Y_i=j\vert \Omega, \beta, X)&=\mathbb{P}\left(\bigcap_{k\neq j}\lbrace Y_{ij}^\star>Y^\star_{ik}\rbrace\rbrace\right)\\
&= \mathbb{P}\left(\bigcap_{k\neq j}\lbrace (X_{ij}-X_{ik})\beta^\star>u_{ik}-u_{ij}\rbrace\rbrace\right)    
\end{align*}
where integration is taken over $u\sim \mathcal{N}(0,\Omega)$, where $u=(u_{ij})$. The above integral is an orthant probability of dimension $J-1$. A yet more challenging case occurs in the presence of panel data. The latter corresponds to sequential choices of an individual in time. We denote those choices by  the subscript $t$. We observe $(j_t)_{t<T}$ for every individual. Integration is now in dimension $T(J-1)$ and takes the form:
\[
L(Y_{i,1:T}=j_{1:T}\vert \Omega, \beta,X)= \mathbb{P}\left(\bigcap_{t=1}^T\bigcap_{k\neq j_t}\lbrace (X_{ij_t t}-X_{ikt})\beta^\star>u_{ikt}-u_{ij_t t}\rbrace\rbrace\right)    
\]

We take the covariance structure studied in \cite{Bursh-Supan1992}. The noise term is $u_{itk}=\alpha_{ik}+\eta_{ikt}$ where $\eta_{ikt}=\varrho_i\eta_{ikt-1}+\nu_{it}$, where $(\alpha_{ik})$ are correlated amongst choices, so are $\nu_{it}$. The terms are all Gaussian.

The dataset is simulated to allow for examples that are more complex, and of variable size. In the model presented above individuals are independent so that we present results in computing the integral for $n=1$, and have already a big advantage of using our methodology. 

\begin{figure}[H]
\begin{center}
\begin{tabular}{cc}
\subfloat[$J=10$, $T=10$]{\includegraphics[scale=0.3]{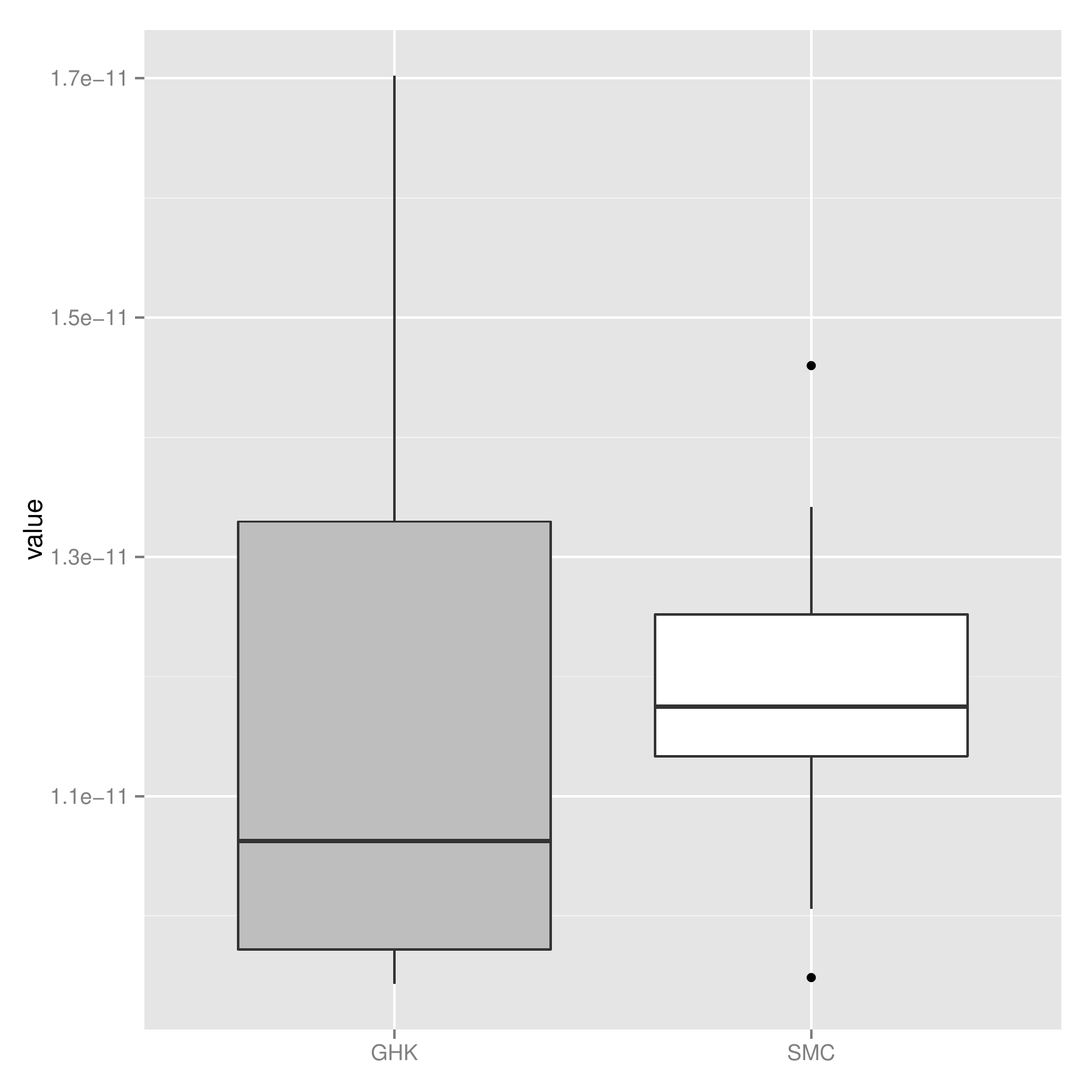}} &
\subfloat[$J=10$, $T=15$]{\includegraphics[scale=0.3]{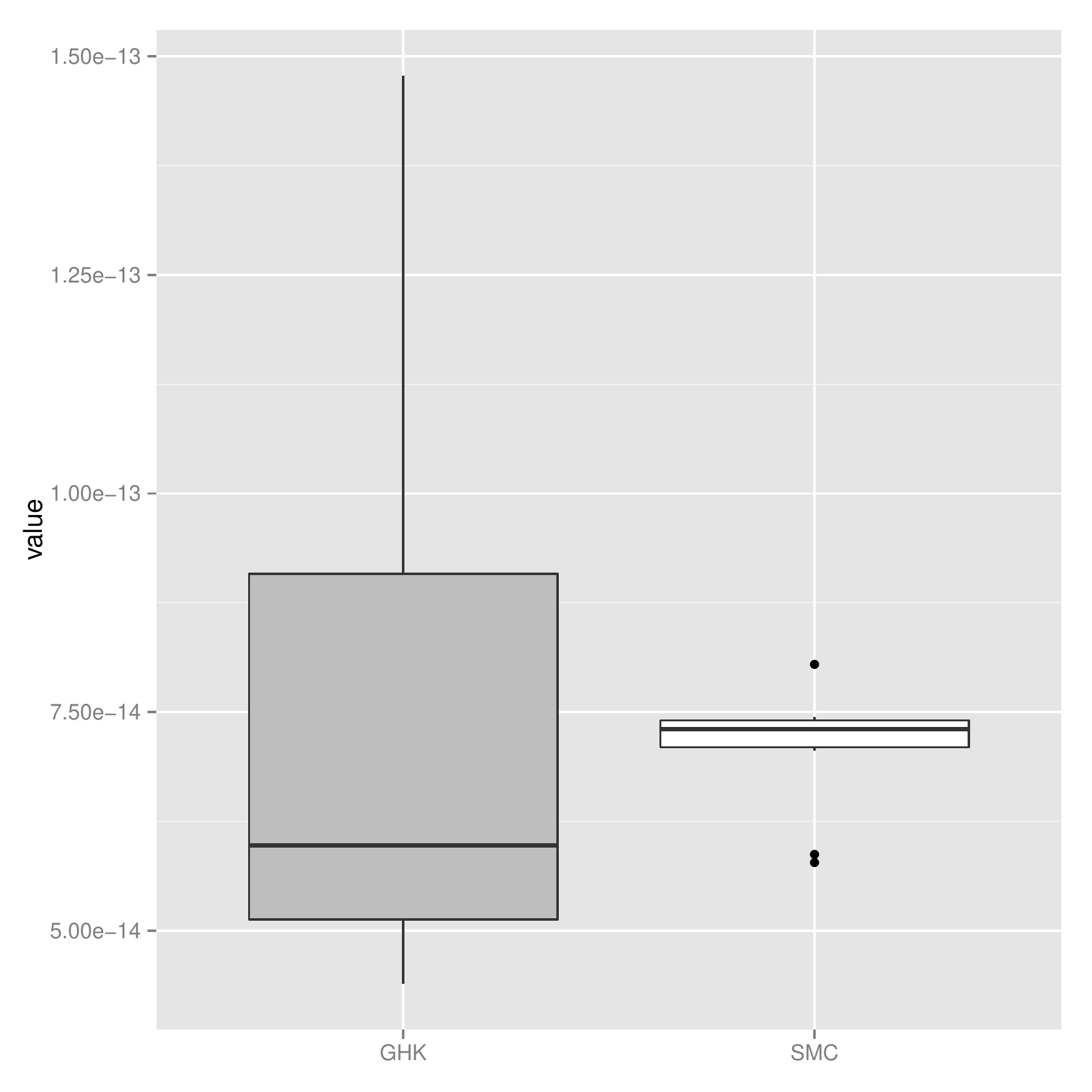}} 
\end{tabular}
\caption{Estimates of the likelihood of a multivariate Probit}
\label{fig:Probit}
\begin{minipage}{15cm}
\footnotesize{Dataset: The Data is simulated using the covariance proposed in \cite{Bursh-Supan1992}, the value of the parameter for which the likelihood is evaluated is taken at random.}
\vspace*{3mm}
\hrule
\end{minipage}
\end{center}
\end{figure}

Figure (\ref{fig:Probit}) shows, for two problems of different size, the gain in precision at constant computational cost. The improvement is substantial and increases with the size of the problem. Taking $n>1$ would only increase this effect as it would consist in taking products of such estimators. 

The latter result suggests that we could use the likelihood for inference using either maximum likelihood or PMCMC. Computing the likelihood with lowest possible variance is also a key issue in finding the evidence in the most precise manner possible.
 
When comparing the two algorithms we set the number of particles of SMC to $M=1000$, for the same computational cost we allocate $M=881031$ to GHK. SMC has however still a lower variance. Regardless of computational time, the ability to compute precise integrals for a small number of particles can also be of importance. It is the case for instance in SMC$^2$ \citep{Chopin2013a} where whole trajectories have to be kept in memory for several samplers.

One could extend these models to multivariate Probit with Student distribution and to multivariate logit models for more robustness. In this case we can use the algorithm based on the mixture representation built in Section \ref{sec:Student}.

\section{Conclusion}
\label{sec:ccl}
We have shown empirically  that the GHK algorithm collapses when the dimension of the problem increases (returning NaN values). In other cases, the distribution of estimates generated by GHK may have heavy tails (see also Remark \ref{rmk:lognorm}). Theoretically for at least one covariance structure we have shown that the variance of the algorithm diverges exponentially fast with the dimension (Section \ref{sec:Markovian}). Our SMC algorithm seems to correct this behavior, and was found to be of practical use for many problems. 

We have tested several kernels as part of the move step (Section \ref{sec:smc}). We advise the practitioners to use Gibbs sampling as the standard ``go to'' move step. However improvements in speed can be achieved for dimensions around 50 using only a partial update. In addition as the dimension increases one might want to use a method with lower complexity at the cost of having to repeat the move a bit more. In this case we recommend the use of an overrelaxed random walk Metropolis-Hastings. 

We have shown that the same idea can be use for computing probabilities of mixtures of Gaussians. In addition we can use the weighted particles returned by the algorithm to compute other integrals (mean, variance, .etc). This approach can outperform a classical Gibbs sampler when the dimension exceeds $20$.

\section*{Acknowledgment} 

This work is part of my PhD thesis under the supervision of Nicolas Chopin. I am grateful for his comments.

\bibliography{biball}
\appendix
\section{Proof of proposition 2.1}
\label{app:proof}
\begin{proof}
	We have that for $0<b<\infty$ the transition density $p^b(x,y)dy$, associated with the kernel $P^b(x,\dd y)$ with respect to the Lebesgue measure, is lower bounded by a constant and the transition is continuous. This Markov chain is a $\psi$-irreducible on a compact support. Hence we can show that the whole support $\left[0,b\right]$ is small \citep{Meyn2009}.

Hence by theorem 16.1.2 to show V-Uniform ergodicity the transition must satisfy the drift condition:
$$
\int p(x,y)V_e(y)dy\leq (1-\beta)V_e(x)+c\ind_{\left[0,b \right]}(x)
$$
for $\beta>0$, $c<\infty$ and a certain $V_e(x)$ with value in $\left[1,\infty\right)$. We take $V_e(x)=ex^2+1$, $e>0$. In the following we check this condition.

The left hand side is given by $\mathbb{E}(X_t^2\vert X_{t-1}=x)$ for the above transition probability,
$$
\mathbb{E}(X_t^2\vert X_{t-1}=x)=\varrho^2 x^2+1+\frac{\varrho\varphi(\varrho x)-b \varphi(b-\varrho x)}{\Phi(b-\varrho x)-\Phi(-\varrho x)}\varrho x 
$$
The ratio is continuous on the bounded set $\left[0,b\right]$, and can be bounded by a constant, such that the by taking $\beta=1-\varrho^2>0$ the drift condition is satisfied for a $c(e)$ depending on $e$.

In addition we can compute exactly the invariant measure. It is unique and given by the solution of:
$$
\pi(y)=\int \frac{\varphi(y;\varrho x,1)}{\Phi(a-\varrho x)-\Phi(-\varrho x)}\ind_{\left[0,b\right]}(y)\pi(x)dx
$$
The solution of the above equation is a truncated skew-Normal distribution,
$$
\pi(dy)\propto\left\lbrace\Phi(b-\varrho y)-\Phi(-\varrho y)\right\rbrace\varphi(y;0,1-\varrho^2)\ind_{\left[0,b\right]}(y)dy.
$$ 
The moments of this distribution have been studied  in \cite{Flecher2009}, in particular note that $0<\mathbb{V}_{\pi}(X)<\infty$.

Define $\psi:x\mapsto\log\left[\Phi(b-\varrho x)-\Phi(-\varrho x)\right]$, by theorem 17.0.1 \citep{Meyn2009} to obtain a CLT for $\frac1{\sqrt{T}}\sum_{t=1}^T\psi(X_t)$ we must ensure that there exist a constant $e>0$ such that $\psi^2(x)<V_e(x)$ on $\left[0,b\right]$. Such a constant can be found by noting that $\psi(x)$ is bounded as long as $b>0$ and that $V_e$ is strictly increasing of $e>0$ with value on $(0,\infty)$. The value of $e$ depends on $b$. We obtain the following convergence result,
$$
\frac1{\sqrt{T}}\left(\sum_{t=1}^T\psi(X_t)-T\mathbb{E}_\pi(\psi(X))\right)\leadsto \mathcal{N}\left(0,\mathbb{V}_\pi\left\lbrace\psi(X)\right\rbrace+\tau\right),
$$
where the variance term is defined because $\psi$ is bounded on $\left[0,b\right]$, and $\tau=2\sum_{k=1}^\infty \text{cov}(X_0,X_k)$.
By taking the exponential and using the continuous mapping theorem (p.7 \cite{VanderVaart1998}) we get a log-normal limiting distribution 
$$
\left(\frac{\prod_{t=1}^T\left(\Phi(b-\varrho X_t)-\Phi(-\varrho X_t)\right)}{\exp\lbrace T\mathbb{E}_\pi\log\left(\Phi(b-\varrho X)-\Phi(-\varrho X)\right)\rbrace}\right)^{\frac1{\sqrt{T}}}\leadsto \mathcal{EN}\left(0,\mathbb{V}_\pi\left\lbrace\psi(X)\right\rbrace+\tau\right).
$$

By Portmanteau's Lemma (p.6 \cite{VanderVaart1998}) for $x^2$ as a continuous and positive function,
\begin{align*}
&\liminf_{T\rightarrow\infty}\mathbb{E}\left[\left(\frac{\prod_{t=1}^T\left(\Phi(b-\varrho X_t)-\Phi(-\varrho X_t)\right)^2}{\exp\lbrace2T\mathbb{E}_\pi\log\left(\Phi(b-\varrho X)-\Phi(-\varrho X)\right)\rbrace}\right)^{\frac1{\sqrt{T}}}\right]>
\exp\left\lbrace\mathbb{V}_\pi\left[\psi(X)\right]+\tau\right\rbrace\\
&\liminf_{T\rightarrow\infty}\mathbb{E}\left[\left(\frac{\prod_{t=1}^T\left(\Phi(b-\varrho X_t)-\Phi(-\varrho X_t)\right)^2}{\exp\lbrace2T\mathbb{E}_\pi\log\left(\Phi(b-\varrho X)-\Phi(-\varrho X)\right)\rbrace}\right)\right]^{\frac1{\sqrt{T}}}>\exp\left\lbrace\mathbb{V}_\pi\left[\psi(X)\right]+\tau\right\rbrace
\end{align*}
where the last line is obtained by Jensen inequality.
The denominator is the square of limit value of the normalizing constant under $x^2$-Uniform ergodicity that follows from the above statement.
\end{proof}
\section{Resampling}
\label{app:resampling}
 \begin{algorithm}[H]
 \caption{Systematic resampling,  $n$ particles}
\begin{algorithmic}
\STATE \underline{Input}: Vector of weights $w$ and vector $x$ to sample from
\STATE \underline{Set} $v\leftarrow n w$, $j\leftarrow 1$, $c=v_1$
\STATE \underline{Sample}: Sample $U\sim \mathcal{U}_{[0,1]}$ 
\FOR{$k=1,\cdots, n$}
\WHILE{$c<u$}
\STATE \underline{Set} $j\leftarrow j+1$, $c\leftarrow c+v_j$
\ENDWHILE
\STATE \underline{Set}: $\hat{x}_k\leftarrow x_j$, $u\leftarrow u+1$
\ENDFOR
\RETURN $\hat{x}$
 \end{algorithmic}
\label{alg:Resample}
\end{algorithm}

\section{Variable Ordering}
\label{app:VO}
\begin{algorithm}[H]
\caption{Variable Ordering}
\begin{algorithmic}
\STATE INIT:  $i_1=\arg \min_{1\leq k\leq T} \Phi\left(\left[\frac{a_k}{\gamma_{kk}},\frac{b_k}{\gamma_{kk}}\right]\right)$
\STATE $\eta_1=\frac1{\Phi\left(\left[\frac{a_{i_1}}{\gamma_{i_1 i_1}},\frac{b_{i_1}}{\gamma_{i_1 i_1}}\right]\right)}\int_{\left[\frac{a_{i_1}}{\gamma_{i_1 i_1}},\frac{b_{i_1}}{\gamma_{i_1 i_1}}\right]} \eta\varphi(\eta)d\eta$
\FOR{$i\in \lbrace2,\cdots,d\rbrace$}
\STATE STEP 1 $i_j=\arg \min_{j\leq k\leq T} \Phi\left(\left[\frac1{\tilde{\gamma}_{kk}}\left\lbrace a_k-\sum_{l=1}^{j-1}\tilde{\gamma}_{i_l k}\eta_k\right\rbrace,\frac1{\tilde{\gamma}_{kk}}\left\lbrace b_k-\sum_{l=1}^{j-1}\tilde{\gamma}_{i_l k}\eta_k\right\rbrace\right]\right)$
\STATE STEP 2 $\eta_j=\frac1{\Phi\left(\left[\frac1{\tilde{\gamma}_{kk}}\left\lbrace a_k-\sum_{l=1}^{j-1}\tilde{\gamma}_{i_l k}\eta_k\right\rbrace,\frac1{\tilde{\gamma}_{kk}}\left\lbrace b_k-\sum_{l=1}^{j-1}\tilde{\gamma}_{i_l k}\eta_k\right\rbrace\right]\right)}\int_{ \left[\frac1{\tilde{\gamma}_{kk}}\left\lbrace a_k-\sum_{l=1}^{j-1}\tilde{\gamma}_{i_l k}\eta_k\right\rbrace,\frac1{\tilde{\gamma}_{kk}}\left\lbrace b_k-\sum_{l=1}^{j-1}\tilde{\gamma}_{i_l k}\eta_k\right\rbrace\right]} \eta\varphi(\eta)d\eta$
\ENDFOR
\RETURN  $(i_1,\cdots,i_d)$ 
\end{algorithmic}
\label{alg:VO}
\end{algorithm}
Where $\tilde{\gamma}$ is updated accordingly when the order is changed.

\section{Hamiltonian Monte Carlo}
\label{app:HMC}
\begin{figure}[H]
\begin{center}
\begin{tabular}{cc}
\subfloat[Dimension 30]{\includegraphics[scale=0.3]{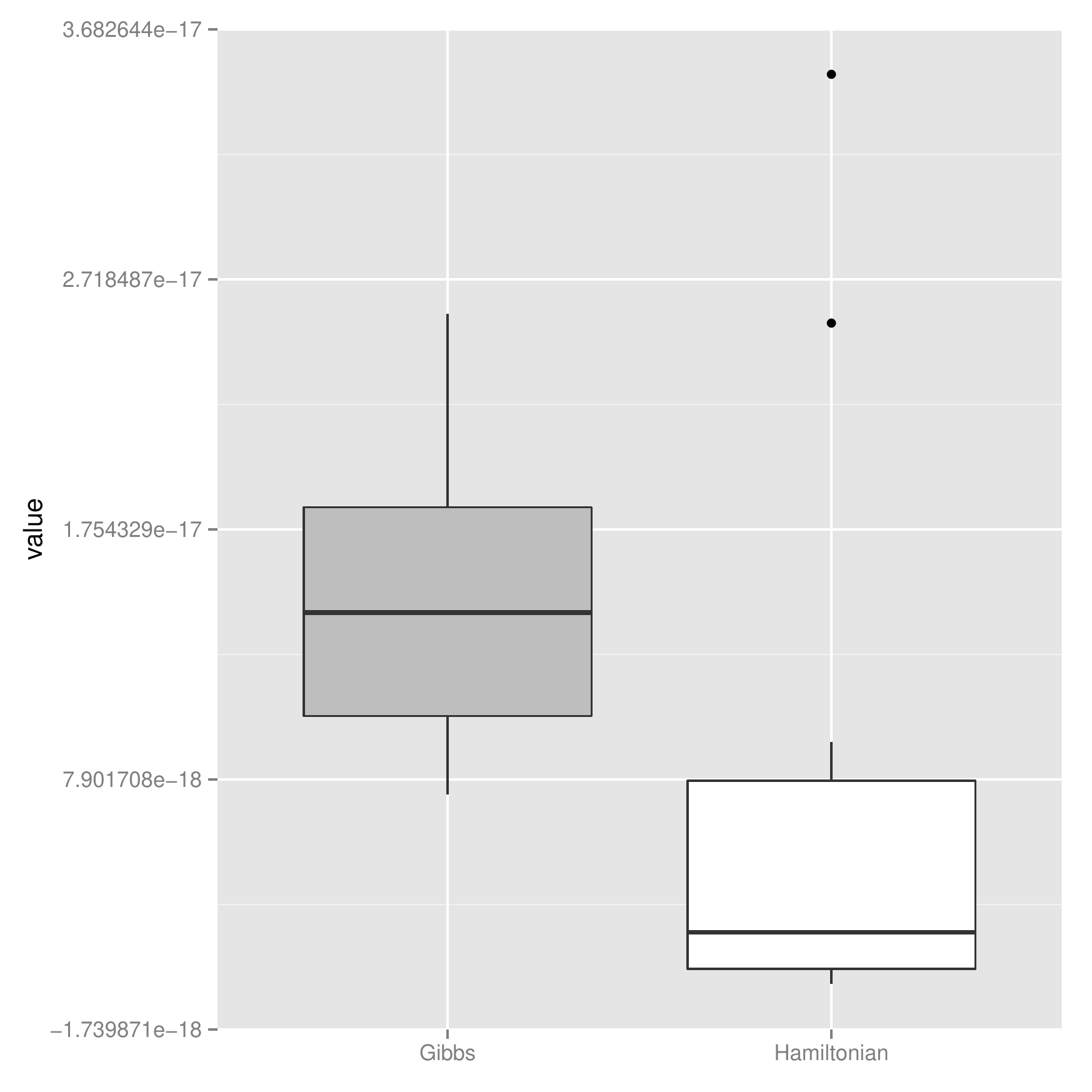}} &
\subfloat[Dimension 40]{\includegraphics[scale=0.3]{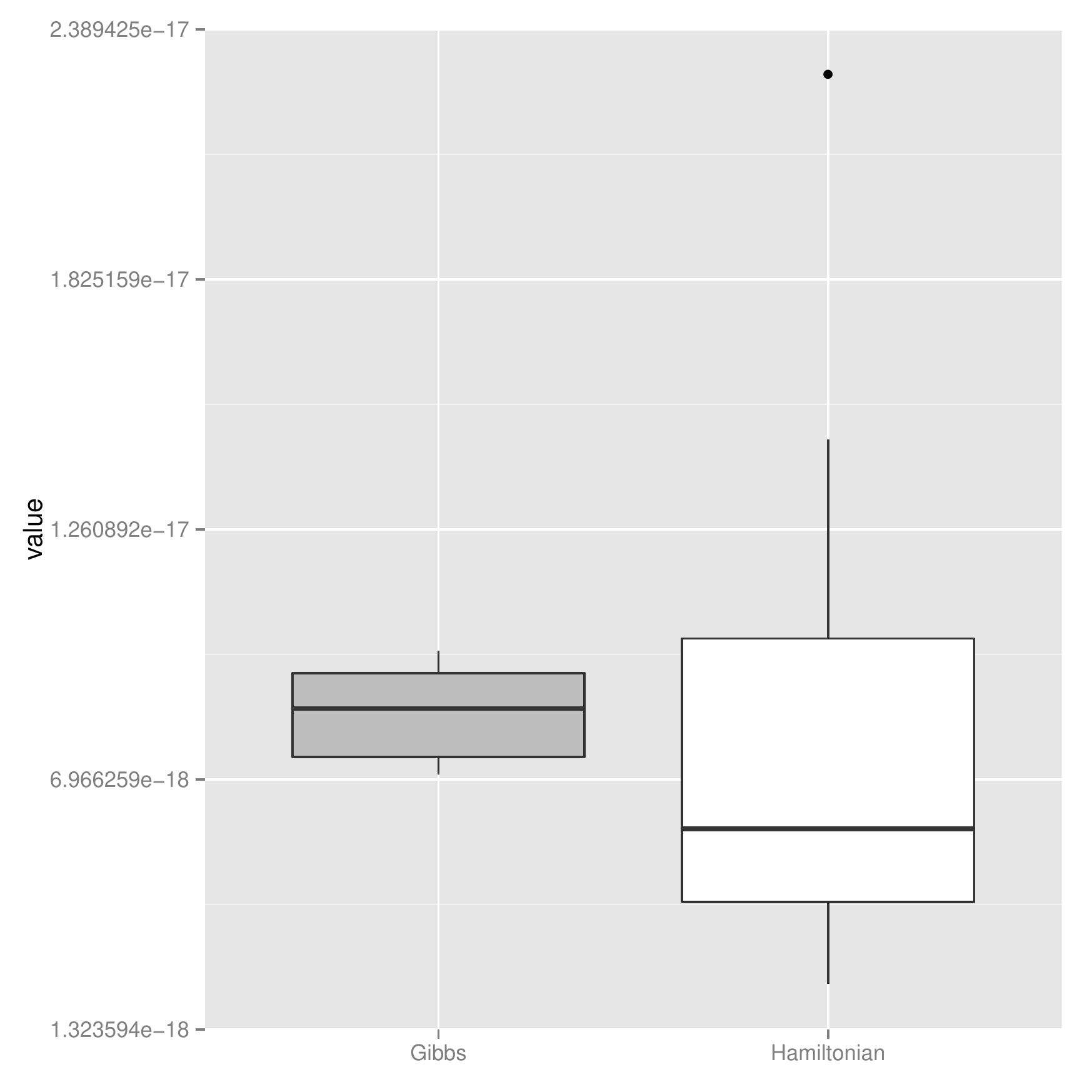} }\\
\subfloat[Dimension 50]{\includegraphics[scale=0.3]{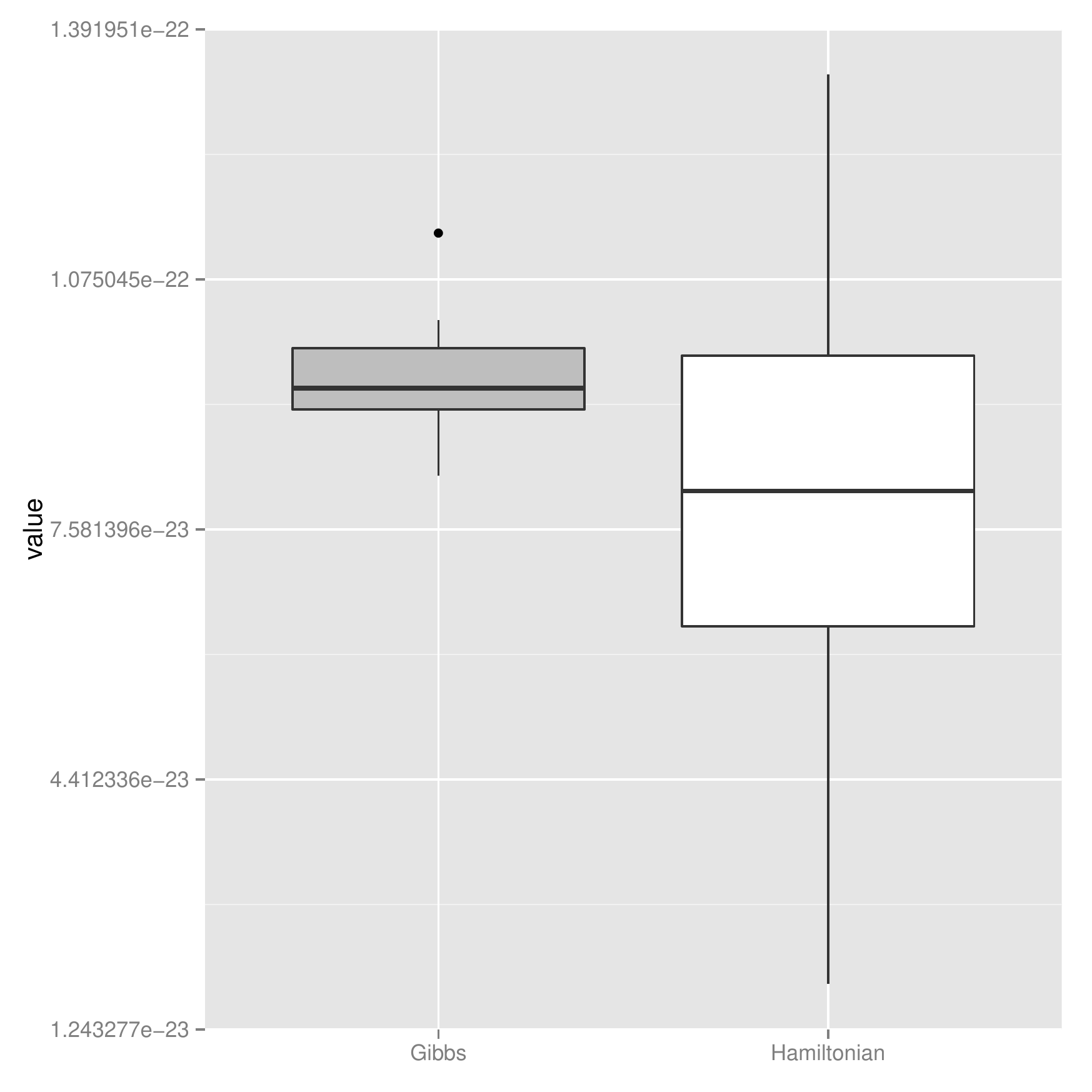} } &
\subfloat[Dimension 60]{\includegraphics[scale=0.3]{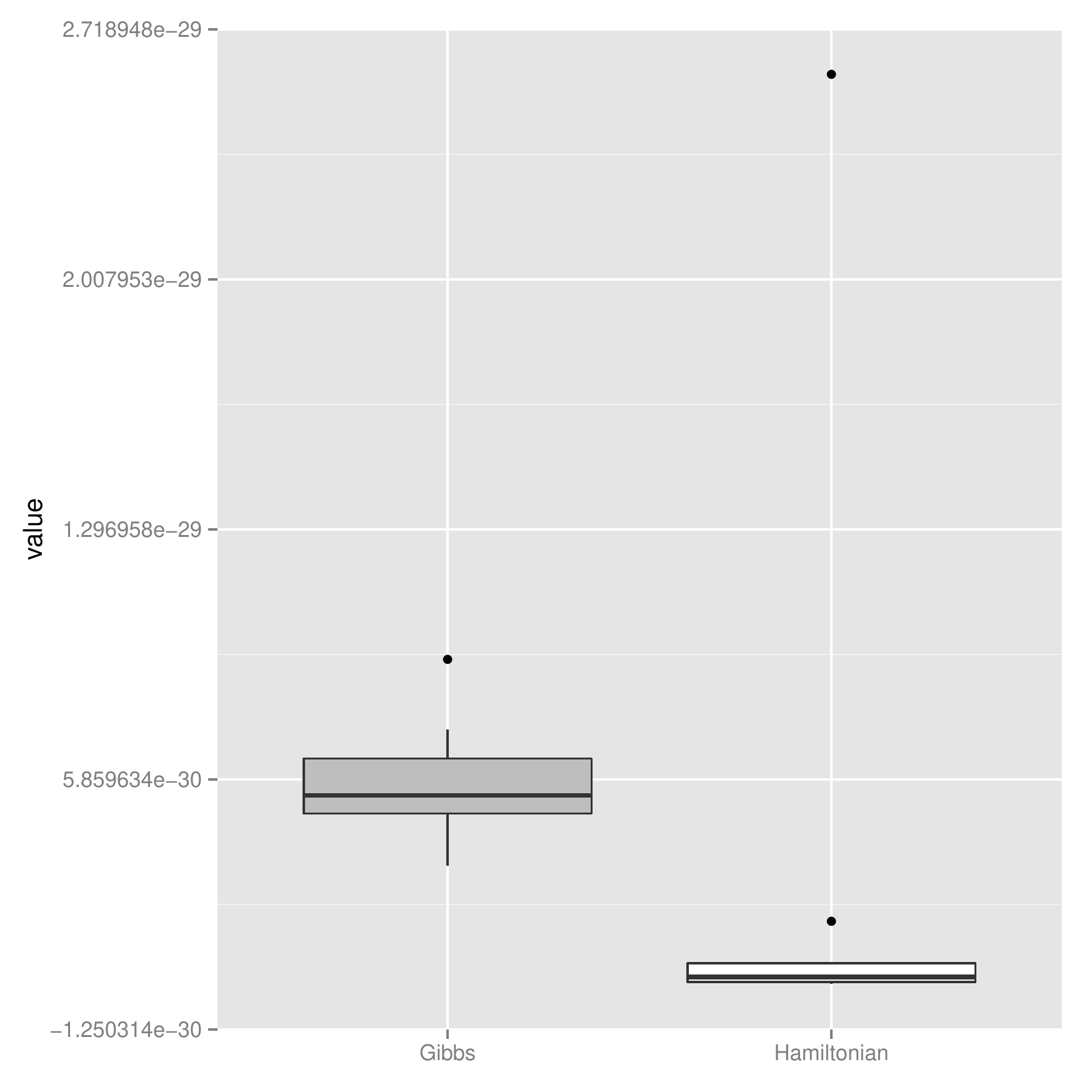} }
\end{tabular}
\caption{Estimates of Orthant probabilities Gibbs vs HMC}
\label{fig:HMC}
\begin{minipage}{15cm}
\footnotesize{Covariance matrices generated from random samples with heavy tails. The various dimension are simulated with the same algorithm and same seed, such that the small ones are subsets of the others. The grey boxplot corresponds to the Gibbs sampler the white to the HMC. The Gibbs sampler seem to have smaller variance and no outliers.}
\vspace*{3mm}
\hrule
\end{minipage}
\end{center}
\end{figure}
\end{document}